\documentclass[12pt,reqno]{amsart}
\usepackage{amsmath,amsfonts,amsthm,amssymb,mathabx,dsfont}
\newcommand\version{April 10, 2020}
\usepackage{amsxtra,hyperref,color,mathtools}

\newtheorem{theorem}{Theorem}[section]
\newtheorem{proposition}{Proposition}[section]
\newtheorem{lemma}{Lemma}[section]
\newtheorem{corollary}{Corollary}[section]

\theoremstyle{definition}

\theoremstyle{remark}
\newtheorem{remark}[theorem]{Remark}
\newtheorem{assumption}{Assumption}

\numberwithin{equation}{section}

\newcommand{\Hh}{\mathbb{H}}
\newcommand{\Aa}{\mathbb{A}}

\newcommand{\eps}{\epsilon}
\renewcommand{\epsilon}{\varepsilon}

\newcommand{\F}{\mathcal{F}}
\newcommand{\E}{\mathcal{E}}

\newcommand{\N}{\mathbb{N}}

\newcommand{\R}{\mathbb{R}}

\newcommand\const{{\rm const.\, }}
\newcommand\pp{{\varphi^{\rm P}}}
\newcommand\id{\mathds{1}}

\DeclareMathOperator{\dist}{dist}
\DeclareMathOperator{\infspec}{inf\, spec}

\DeclareMathOperator{\ran}{ran}

\DeclareMathOperator{\Tr}{Tr}
\DeclareMathOperator{\tr}{Tr}
\DeclareMathOperator{\spa}{span}

\begin{document}

\title[Quantum corrections for a strongly coupled polaron --- \version]{Quantum corrections to the Pekar asymptotics of a strongly coupled polaron}

\author{Rupert L. Frank}
\address[R. Frank]{Mathematisches Institut, Ludwig-Maximilians Universit\"at M\"unchen, The\-resienstr. 39, 80333 M\"unchen, Germany, and Department of Mathematics, California Institute of Technology, Pasadena, CA 91125, USA}
\email{r.frank@lmu.de, rlfrank@caltech.edu}

\author{Robert Seiringer}
\address[R. Seiringer]{IST Austria (Institute of Science and Technology Austria), Am Campus 1, 3400 Klosterneuburg, Austria}
\email{robert.seiringer@ist.ac.at}

\begin{abstract}
We consider the Fr\"ohlich polaron model in the strong coupling limit. 
It is well known that to leading order the ground state energy is given by the (classical) Pekar energy. In this work, we establish the subleading correction, describing quantum fluctuation about the classical limit. Our proof applies to a model of a confined polaron, where both the electron and the polarization field are restricted to a set of finite volume, with linear size determined by the natural length scale of the Pekar problem. 
\end{abstract}

\renewcommand{\thefootnote}{${}$} \footnotetext{\copyright\, 2019 by
  the authors. This paper may be reproduced, in its entirety, for
  non-commercial purposes.}
  
\date{\version}

\maketitle



\section{Introduction}

The polaron model was introduced by Fr\"ohlich  \cite{F37} as a model of an electron interacting with the quantized optical modes of a polar crystal. It represents a simple and well-studied model of non-relativistic quantum field theory, and we refer to \cite{AD,FLST,GL,JSM,Spohn} for properties, results and further references.  

In the strong coupling limit $\alpha\to \infty$, the model allows for an exact solution, in the sense that the ground state energy asymptotically equals the one given by the Pekar approximation \cite{P54}, which amounts to a classical approximation to the quantum field theory. This was first shown by Donsker and Varadhan \cite{DoVa} using a path integral formulation of the problem. (See also \cite{MV1,MV2} for recent work on the construction of the Pekar process \cite{Spohn}.) Later the result was improved by Lieb and Thomas \cite{LT} who provided a quantitative bound on the difference. 

We are interested here in the subleading correction to the classical (Pekar) approximation. It was predicted in the physics literature (see \cite{Allcock,Alcock,tja,Gross} and references there) that this correction results from quantum fluctuations about the classical limit, and is $O(\alpha^{-2})$ smaller than the  main term. It can be calculated by evaluating the ground state energy of a system of (infinitely many) harmonic oscillators with frequencies determined by the Hessian of the Pekar functional. This result is verified rigorously in this paper, by giving upper and lower bounds on the ground state energy of the Fr\"ohlich polaron model that establish this subleading  correction.  Our analysis applies to a model of a confined polaron, where both the electron and the polarization field are restricted to a finite volume (with linear size of the natural length scale set by the Pekar problem). 

The confinement breaks translation invariance, which removes zero modes otherwise present in the Hessian of the Pekar functional, and avoids having to localize the electron on the Pekar scale, which simplifies the problem. The singular ultraviolet behavior is unaffected by the confinement, however, and represents one of the main technical challenges. A key ingredient in our analysis is a multiple use of the commutator method of Lieb and Yamazaki \cite{LY}, combined with Nelson's Gross transformation \cite{grossT,nelson}.

\section{Model and Main Results}

\subsection{The Model}
For $\Omega\subset \R^3$ open, let $\Delta_\Omega$ denote the Dirichlet Laplacian, and let $v_x(\,\cdot\,) = (-\Delta_\Omega)^{-1/2}(x,\,\cdot\,)$. The model we consider is defined by the Hamiltonian
\begin{equation}\label{def:ham}
\Hh \coloneqq -\Delta_\Omega - a(v_x) - a^\dagger(v_x) + \N
\end{equation}
in $L^2(\Omega)\otimes\mathcal F$, where $\mathcal F$ is the bosonic Fock space over $L^2(\Omega)$. The creation and annihilation operators satisfy the commutation relation
\begin{equation}
[a(f),a^\dagger(g)] = \alpha^{-2} \langle f| g\rangle
\qquad\text{for}\ f,g\in L^2(\Omega)
\end{equation}
with a parameter $\alpha>0$. The field energy is given by the number operator $\N = \sum_j a^\dagger(\varphi_j)a(\varphi_j)$ for some orthonormal basis $\{ \varphi_j\}$ in $L^2(\Omega)$, with spectrum $\sigma(\N) = \alpha^{-2} \{0,1,2,\ldots\}$. We are interested in the ground state energy of $\Hh$ as $\alpha\to \infty$.

We note that the expression \eqref{def:ham} is somewhat formal, since $v_x\not\in L^2(\Omega)$ and hence $a^\dagger(v_x)$ is not densely defined. The operator $\Hh$ can be defined with the aid of its corresponding quadratic form, however. It is in fact well known that $\Hh$ defines a self-adjoint operator on a suitable domain, see \cite{griesemer} or Section~\ref{section:GT} below. 

\begin{remark}
By rescaling all lengths by $\alpha$, $\Hh$ is unitarily equivalent to the operator $\alpha^{-2} \tilde \Hh$ with
\begin{equation}
\tilde \Hh = -\Delta_{\Omega/\alpha} - \sqrt{\alpha}\left( \tilde a(\tilde v_x) - \tilde a^\dagger(\tilde v_x) \right) + \tilde \N
\end{equation}
where $\tilde v_x(\,\cdot\,) = (-\Delta_{\Omega/\alpha})^{-1/2}(x,\,\cdot\,)$, $\tilde\N = \sum_j \tilde a^\dagger(\varphi_j)\tilde a(\varphi_j)$ and the $\tilde a$ and $\tilde a^\dagger$ operators satisfy $[\tilde a(f),\tilde a^\dagger(g)] = \langle f| g\rangle$ (and are thus independent of $\alpha$). Large $\alpha$ hence corresponds to the strong-coupling limit of a polaron confined to a region of linear size $\alpha^{-1}$. We find it more convenient to work in the variables defined in \eqref{def:ham}, however.
\end{remark}

\begin{remark}
Typically the polaron model is considered without confinement, i.e., for $\Omega=\R^3$, in which case the electron-phonon coupling function equals  $(-\Delta_{\R^3})^{-1/2}(x, y) = (2\pi)^{-3} \int_{\R^3} e^{ik\cdot (x-y)} |k|^{-1} dk = (2\pi^2)^{-1} |x-y|^{-2}$. For the proof of our main theorem the compactness of $(-\Delta_{\Omega})^{-1}$ will be important, hence we need to consider bounded sets $\Omega$ here.
\end{remark}

\subsection{Pekar Functional(s)}

We introduce the  classical energy functional corresponding to \eqref{def:ham} as 
\begin{equation}\label{epv}
\E(\psi,\varphi) = \int_\Omega |\nabla \psi(x) |^2 \, dx   - 2 \iint\nolimits_{\Omega\times\Omega} \varphi(x) (-\Delta_\Omega)^{-1/2}(x,y) |\psi(y)|^2 \,dx\,dy  + \int_\Omega \varphi(x)^2\, dx
\end{equation}
where $\psi \in H_0^1(\Omega)$, $\|\psi\|_2=1$, and  $\varphi \in L^2_\R(\Omega)$, the  real-valued functions in $L^2(\Omega)$. Formally, it can be obtained from \eqref{def:ham} by replacing the field operators $a(f)$ and $a^\dagger(f)$ by $\int \varphi(x) f(x) dx$, and taking an expectation value with the electron wave function $\psi$. 
The Pekar energy is
\begin{equation}\label{def:ep}
e^{\rm P} = \min_{\psi,\varphi} \E(\psi,\varphi) \,.
\end{equation}
For $\Omega=\R^3$ it was shown in \cite{DoVa,LT} that $\infspec \Hh\to e^{\rm P}$ as $\alpha\to\infty$. The result can be shown to hold also for general $\Omega$. Our goal here is to compute the subleading correction in this asymptotics.

We will work under the following

\begin{assumption}\label{ass:unique}
The functional $\E$ in \eqref{epv} has a unique minimizer $\psi^{\rm P}, \varphi^{\rm P}$ (up to a trivial constant phase factor for $\psi^{\rm P}$).
\end{assumption}

Our proof works under the more general assumption that the set of minimizers of $\E$ is discrete (up to the phase degeneracy). The case where minimizers form a  continuous manifold requires additional ideas, however.

Since $\E(|\psi|,\varphi)\leq\E(\psi,\varphi)$
we assume from now on that $\psi^{\rm P}$ is non-negative.
For given $\psi$, the choice of the minimizing $\varphi$ is clearly unique, and vice versa. In particular, our  Assumption~\ref{ass:unique} concerns uniqueness of the minimizer of the corresponding Pekar functional
\begin{equation}
\E^{\rm P}(\psi) = \min_\varphi \E(\psi,\varphi) = \int_\Omega |\nabla \psi(x) |^2 \, dx - \iint\nolimits_{\Omega\times\Omega} |\psi(x)|^2 (-\Delta_\Omega)^{-1}(x,y) |\psi(y)|^2 \,dx\,dy\,.
\end{equation}
Recall that, for $\Omega=\R^3$, uniqueness of minimizers of $\E^{\rm P}$ (up to translations and phase factor) is known \cite{Lieb} (see also \cite{Tod}).  We expect Assumption~\ref{ass:unique} to hold if $\Omega$ is convex, for instance. The proof in \cite{Lieb}  can be adapted to show uniqueness in case $\Omega$ is a ball~\cite{dario}.

\begin{assumption}\label{ass:hessian}
There exists a $\kappa>0$ such that 
\begin{equation}\label{ass:hp}
\E^{\rm P}(\psi) \geq \E^{\rm P}(\psi^{\rm P}) +  \kappa \int_{\Omega} \left| \nabla \left(\psi - \psi^{\rm P} \right)\right|^2  \quad \forall \psi \in H_0^1(\Omega), \, \psi\geq 0, \, \|\psi\|_2=1\,.
\end{equation}
\end{assumption}

The bound  \eqref{ass:hp} follows from an a priori weaker spectral assumption on the absence of non-trivial zero modes of the Hessian of $\E^{\rm P}$ at its minimizer $\psi^{\rm P}$, by a simple compactness argument. For completeness, we spell out the details of this argument in Appendix~\ref{sec:compact}. 
 The analogue of this spectral assumption in the case $\Omega=\R^3$  is known (up to zero-modes resulting from the translation invariance) \cite{Lenzmann,WeWi}.
Using the method in \cite{Lenzmann}, one can prove Assumption~\ref{ass:hessian} in case $\Omega$ is a ball \cite{dario}. 

If one minimizes $\E(\psi,\varphi)$ over $\psi$ for given $\varphi$, one obtains the functional
\begin{equation}\label{def:pekarF}
\F^{\rm P}(\varphi) = \min_\psi \E(\psi,\varphi) = \|\varphi\|_2^2 +  \infspec \left( -\Delta_\Omega + V_\varphi(x)  \right) 
\end{equation}
where $V_\varphi = -2 (-\Delta_\Omega)^{-1/2} \varphi$. 
Let $H^{\rm P}$ denote its Hessian at the unique minimizer $\varphi^{\rm P}$, i.e.,
\begin{equation}\label{def:hessian}
\lim_{\eps\to 0} \frac 1{\eps^2} \left( \F^{\rm P}(\varphi^{\rm P} + \eps \varphi) - e^{\rm P} \right) = \langle \varphi | H^{\rm P} | \varphi\rangle \quad \forall \varphi\in L^2_{\R}(\Omega). 
\end{equation}
An explicit computation gives 
\begin{equation}\label{def:hp}
H^{\rm P} = \id  - 4 (-\Delta_\Omega)^{-1/2} \psi^{\rm P} \frac {Q^{\rm P}}{ -\Delta_\Omega + V_{\varphi^{\rm P}}  -\mu^{\rm P}  } \psi^{\rm P} (-\Delta_\Omega)^{-1/2}
\end{equation}
where $\psi^{\rm P}$ acts as a multiplication operator,  $\mu^{\rm P} = \infspec (-\Delta_\Omega + V_{\varphi^{\rm P}}) = e^{\rm P} -\|\varphi^{\rm P}\|_2^2$, and $Q^{\rm P}$ is the projection orthogonal to $\psi^{\rm P}$, i.e., orthogonal to the kernel of $-\Delta_\Omega + V_{\varphi^{\rm P}}-\mu^{\rm P}$. 
It is not difficult to see that Assumption \ref{ass:hessian} implies that  $H^{\rm P}$ is non-degenerate, i.e., strictly positive (compare with Proposition~\ref{hessianglobal} in Section~\ref{ss:hessianglobal} below).

Finally, we need a regularity assumption on the domain $\Omega$.

\begin{assumption}\label{ass:domain}
The domain $\Omega$ is bounded, and has a $C^{3,\delta}$ boundary for some $0<\delta<1$.
\end{assumption}

For a proper definition of the meaning of $C^{3,\delta}$ boundary, see Appendix~\ref{sec:appa}. Assumption~\ref{ass:domain} allows to estimate derivatives of the integral kernel of certain functions of the Dirichlet Laplacian (see Appendix~\ref{sec:appb}). The required estimates certainly hold under less restrictive assumptions on $\Omega$, and we expect our main result to hold also in case $\Omega$ is a cube, for instance. We shall not try to investigate the minimal regularity assumptions, however, and shall henceforth work with Assumption~\ref{ass:domain}.


\subsection{Main Result}\label{ss:mr}

Recall the definition \eqref{def:ep} for the Pekar energy $e^{\rm P}$, as well as (\ref{def:hp}) for the Hessian $H^{\rm P}$ of $\F^{\rm P}$ in \eqref{def:pekarF} at the unique minimizer $\varphi^{\rm P}$. 
Our main result is as follows. 

\begin{theorem}\label{main:thm}
Under Assumptions \ref{ass:unique}--\ref{ass:domain} one has, as $\alpha\to \infty$, 
\begin{equation}\label{eq:thm}
\infspec \Hh = e^{\rm P} -  \frac 1 {2\alpha^{2}} \Tr \left( \id - \sqrt{ H^{\rm P}} \right)  + o(\alpha^{-2})\,.
\end{equation}
More precisely,  the bounds 
\begin{equation}\label{eq:thm2}
 -C \alpha^{-1/7} (\ln \alpha)^{5/14}  \leq \alpha^2 \infspec \Hh - \alpha^2 e^{\rm P} +  \frac 1{2} \Tr \left( \id - \sqrt{ H^{\rm P}} \right) \leq C \alpha^{-2/11}
\end{equation}
hold for some constant $C>0$  and $\alpha$ large enough.
\end{theorem} 

The trace in \eqref{eq:thm} and \eqref{eq:thm2} is over $L^2(\Omega)$. We shall see below that $\id - \sqrt{H^{\rm P}}$ is actually trace class. Note also that $H^{\rm P}< \id$, hence the coefficient of $\alpha^{-2}$ in \eqref{eq:thm} is strictly negative.

In the case $\Omega=\R^3$, the correctness of the leading term $e^{\rm P}$ was shown in \cite{DoVa,LT}. The proof in \cite{LT} gives an error bound of the order $\alpha^{-1/5}$. In the confined case considered here, we improve this error bound to $O(\alpha^{-2})$, and actually compute the next order correction. We conjecture that the formula  \eqref{eq:thm} also holds true in case $\Omega=\R^3$, as predicted in the physics literature \cite{Allcock,Alcock,tja,Gross}. Our upper bound, in fact, can easily be generalized to this case. While our methods are not strong enough to prove the corresponding lower bound, parts of our proof are applicable also to the $\Omega=\R^3$ case, and yield an improved error bound compared to the one given in \cite{LT}.

The $\alpha^{-2}$ correction to the ground state energy in \eqref{eq:thm} can be interpreted as arising from quantum fluctuations around the classical limit described by the Pekar functional. The trace originates from the ground state energy of a Hamiltonian describing a system of (infinitely) many harmonic oscillators.

The remainder of the paper is devoted to the proof of Theorem~\ref{main:thm}. We start with a brief outline to guide the reader.

\subsection{Outline of the Proof} 
In Section~\ref{sec:pekar} we study the Pekar functional \eqref{def:pekarF}. We shall compute its Hessian at the unique minimizer $\varphi^{\rm P}$, and use it to estimate the functional in a small neighborhood of its minimizer. We shall also derive a useful quadratic lower bound that is valid globally, i.e., not just close to the minimizer. 

In Section~\ref{sec:upper} we shall derive an upper bound on the ground state energy of $\Hh$ that has the desired asymptotic form as $\alpha\to \infty$. We shall construct an appropriate trial state and utilize the estimate of the Pekar functional close to its minimizer from the previous section. 

Sections~\ref{sec:LY} and~\ref{section:GT} contain auxiliary results that are essential for the lower bound, in particular to allow for an ultraviolet regularization of the problem. In Section~\ref{sec:LY} the commutator method of Lieb and Yamazaki \cite{LY} is applied multiple (in fact, three) times in order to estimate the effect of an ultraviolet cutoff in the coupling function $v_x$ in terms of the  number operator $\N$ and the electron kinetic energy $-\Delta_\Omega$. The relevant operator that needs to be bounded is $\N^{1/2} (-\Delta_\Omega)^{3/2}$, which cannot be controlled in terms of $\Hh^2$, however. The necessary bound does hold after a unitary Gross transformation, which shall be explained in Section~\ref{section:GT}. This will be sufficient for our purpose. 

In Section~\ref{sec:lower} we shall give a lower bound on the ground state energy of $\Hh$ of the desired asymptotic form. We shall use the results of Sections~\ref{sec:LY} and~\ref{section:GT} to implement an ultraviolet cutoff, which effectively reduces the problem to finitely many modes. We shall then use an IMS localization in Fock space and the bounds in Section~\ref{sec:pekar} to conclude the desired lower bound. 

In Appendix~\ref{sec:compact} we shall give an equivalent formulation of Assumption~\ref{ass:hessian} in terms of spectral properties of the Hessian of $\E^{\rm P}$. 
In further appendices we shall derive bounds on derivatives of the integral kernel of certain functions of the Dirichlet Laplacian $\Delta_\Omega$ that we need in our proof. These bounds are derived in Appendix~\ref{sec:appb} utilizing a theorem in Appendix~\ref{sec:appa} on bounds on solutions of Poisson's equation. 

Throughout the proof, we shall use the symbol $a\lesssim b$ if $a\leq C b$ for some constant $C>0$. 


\section{The Pekar Functional}\label{sec:pekar}

\subsection{Hessian of the Pekar Functional}

We consider the Pekar functional \eqref{def:pekarF} and write it as 
\begin{equation}
\F^{\rm P}(\varphi)=e(\varphi)+\|\varphi\|_2^2
\end{equation}
with
\begin{equation}
e(\varphi) = \infspec  H_\varphi \qquad \text{and} \qquad H_\varphi = -\Delta_\Omega + V_\varphi(x) \,.
\end{equation}
Recall that for $\varphi\in L^2_\R(\R^3)$ we set $V_\varphi = -2 (-\Delta_\Omega)^{-1/2} \varphi$. In this section we work under Assumption~\ref{ass:unique} which states that  $\F^{\rm P}(\varphi)$ has a unique minimizer $\pp$. We have $e(\varphi)+\|\varphi\|_2^2 \geq e(\pp) + \|\pp\|_2^2$ and our goal in this section is to obtain upper and lower bounds on the difference.

Recall that $\psi^{\rm P}$ denotes the unique non-negative minimizer of $\E^{\rm P}$, which is the ground state of $H_{\pp}$. We have 
\begin{equation}\label{eq:rvp}
\varphi^{\rm P}=(-\Delta_\Omega)^{-1/2}|\psi^{\rm P}|^2 \,.
\end{equation}
For later use, we record that $\psi^{\rm P}$ is a bounded function.

\begin{lemma}\label{psibounded}
$\psi^{\rm P}\in L^\infty(\Omega)$
\end{lemma}

\begin{proof}
The Euler--Lagrange equation for $\psi^{\rm P}$ reads $-\Delta_\Omega\psi^{\rm P} - 2((-\Delta_\Omega)^{-1}|\psi^{\rm P}|^2)\psi^{\rm P} = \mu\psi^{\rm P}$ for some $\mu\in\R$, which we rewrite as
\begin{equation}
\psi^{\rm P} = (-\Delta_\Omega)^{-1}\left( \left( \mu + 2((-\Delta_\Omega)^{-1}|\psi^{\rm P}|^2) \right)\psi^{\rm P}\right).
\end{equation}
From \eqref{eq:b2} we deduce that $ (-\Delta_\Omega)^{-1}(x,y) \leq (-\Delta_{\R^3})^{-1}(x,y) =(4\pi |x-y|)^{-1}$. 
 By Sobolev's inequality  $|\psi^{\rm P}|^2\in L^3(\Omega)$, and hence by H\"older's inequality  $(-\Delta_\Omega)^{-1}|\psi^{\rm P}|^2\in L^\infty(\Omega)$. Thus, $f=\left( \mu + 2((-\Delta_\Omega)^{-1}|\psi^{\rm P}|^2) \right)\psi^{\rm P}\in L^2(\Omega)$ and once again by H\"older's inequality, $\psi^{\rm P} = (-\Delta_\Omega)^{-1} f\in L^\infty(\Omega)$, as claimed.
\end{proof}

Let $P=|\psi^{\rm P}\rangle\langle\psi^{\rm P}|$ and $Q=\id-P$. We introduce the following non-negative operators 
\begin{equation}\label{def:K}
K = 4  (-\Delta_\Omega)^{-1/2} \psi^{\rm P} \frac Q{ H_{\pp} -e(\pp)  } \psi^{\rm P}  (-\Delta_\Omega)^{-1/2} 
\end{equation}
and
\begin{equation}
L =  4(-\Delta_\Omega)^{-1/2} \psi^{\rm P} ( -\Delta_\Omega)^{-1} \psi^{\rm P} (-\Delta_\Omega)^{-1/2}
\end{equation}
where $\psi^{\rm P}$ acts as a multiplication operator. 
We shall see that $K = \id - H^{\rm P}$, where $H^{\rm P}$ denotes the Hessian of $\F^{\rm P}(\varphi)$ at $\varphi=\varphi^{\rm P}$, introduced in \eqref{def:hessian} above. 

It is easy to see that $L$ is trace class, since $(-\Delta_\Omega)^{-1/2} \psi ( -\Delta_\Omega)^{-1/2}$ is Hilbert--Schmidt for any multiplication operator $\psi \in L^2(\Omega)$. In fact, since 
$(-\Delta_\Omega)^{-1/2} \leq \sqrt{2} (-\Delta_\Omega+ e_1)^{-1/2}$ (with $e_1 = \infspec (-\Delta_\Omega)>0$) and 
$(-\Delta_\Omega+ e_1)^{-1/2}(x,y) \leq (-\Delta_{\R^3}+e_1)^{-1/2}(x,y)$ for any $x,y\in\R^3$ by \eqref{eq:b2},  the Cauchy--Schwarz inequality implies that 
\begin{equation}\label{argu:hs}
\tr \, [ (-\Delta_\Omega)^{-1/2} \psi ( -\Delta_\Omega)^{-1/2}]^2 \leq  \frac 1{(2\pi)^3} \int_{\R^3} \left( \frac 2 {k^2 + e_1} \right)^2 dk \, \int_\Omega |\psi(x)|^2\, dx \,.
\end{equation}
To show that also  $K$ is trace class, we shall first prove the following lemma, which implies, in particular, that $V_\varphi$ is operator-bounded relative to $-\Delta_\Omega$ if $\varphi\in L^2(\Omega)$. 

\begin{lemma} \label{lem:rb}
With $V_\varphi(x) = -2 (-\Delta_\Omega)^{-1/2} \varphi (x)$, 
we have
\begin{equation}
\left\| V_\varphi (-\Delta_\Omega)^{-1} \right\|^2  
\lesssim \langle \varphi| (-\Delta_\Omega)^{-1} |\varphi \rangle \,.
\end{equation}
\end{lemma}

\begin{proof}
Note that the right side is simply the square of the $L^2$-norm of $V_\varphi$. By arguing as in \eqref{argu:hs} one readily checks that the desired bound even holds with the Hilbert--Schmidt norm on the left side.
%
\end{proof}

A straightforward modification of the proof shows that $V_\varphi$ is actually infinitesimally operator-bounded relative to $\Delta_\Omega$, i.e., $\lim_{\kappa\to \infty} \| V_\varphi (-\Delta_\Omega + \kappa)^{-1} \| = 0$. This readily implies that $(-\Delta_\Omega)^{1/2}  \frac Q{ H_{\pp} -e(\pp)  }  (-\Delta_\Omega)^{1/2}$ is bounded, hence the trace class property of $K$ follows from the one of $L$.

Our main result in this section is the following.

\begin{proposition}\label{hessian}
Assume that $\varphi\in L^2_\R(\Omega)$ is such that
\begin{equation}\label{eq:cc}
\|(-\Delta_\Omega)^{-1/2}(\varphi - \pp)\|_2 \leq  \eps
\end{equation}
for  $\epsilon>0$ small enough. Then
\begin{equation}\label{lbh}
\left|  \F^{\rm P}(\varphi)  - \F^{\rm P}(\pp)  -    \langle \varphi- \pp |  \id - K  | \varphi - \pp \rangle\right | 
\lesssim\eps\,  \langle \varphi- \pp |   L | \varphi - \pp \rangle\,.
\end{equation}
\end{proposition}

This result implies, in particular, that $0\leq K\leq \id$. It identifies $H^P= \id -K$ as the Hessian of $\F^{\rm P}(\varphi) = e(\varphi) + \|\varphi\|_2^2$ at the minimizer $\pp$. Our assumption on the strict positivity of the  Hessian thus translates, in view of the compactness of $K$, to the statement $\|K\|<1$.

\begin{proof}
By choosing $\eps>0$ small enough and arguing as in the proof of Lemma \ref{lem:rb} we can ensure that the family of operators $-\Delta_\Omega + V_\varphi(x)$ has a unique eigenvalue close to $e(\pp)$ and this eigenvalue is $e(\varphi)$. The rest of the spectrum of $H_\varphi$ is uniformly bounded away from $e(\pp)$. Hence we can write
\begin{equation}\label{eq:ci}
e(\varphi) =  \Tr \int_C \frac z{z- H_\varphi} \frac{dz}{2\pi i}
\end{equation}
for a fixed (i.e., $\varphi$-independent) contour $C$ that encircles $e(\pp)$. 

We claim   that the operator $\Delta_\Omega(z-H_\pp)^{-1}$ is uniformly bounded for $z\in C$. This follows from the fact that the multiplication operator $V_{\varphi^{\rm P}}$ is infinitesimally operator-bounded relative to $-\Delta_\Omega$, as already argued after the proof of  Lemma~\ref{lem:rb} above. Consequently,  
\begin{equation}\label{3.12}
\sup_{z\in C} \left\| V_{\varphi-\pp} (z-H_\pp)^{-1} \right\| < 1
\end{equation}
for small $\epsilon$, by Lemma~\ref{lem:rb} and our assumption \eqref{eq:cc}. We can thus use the resolvent identity in the form
\begin{align}\nonumber
\frac 1{z- H_\varphi} & = \left( \id - \frac Q{z-H_\pp} V_{\varphi-\pp} \right)^{-1} \frac Q{z- H_\pp}  \\ & \quad + \left( \id - \frac Q{z-H_\pp} V_{\varphi-\pp} \right)^{-1}  \frac P{z- e(\pp)} \left( \id -  V_{\varphi-\pp} \frac 1{z-H_\pp} \right)^{-1} \,.
\end{align}
The first term on the right side is analytic in $z$ for all $z$ inside the contour $C$, and hence gives zero after integration when inserted in \eqref{eq:ci}. The second term is rank one, and Fubini's theorem implies that we can interchange the trace and the integral after inserting this term in \eqref{eq:ci}. We thus obtain
\begin{align}\nonumber 
& e(\varphi) \\ &  =   \int_C \frac z{z- e(\pp)} \left\langle \psi^{\rm P} \left|  \left( \id -  V_{\varphi-\pp} \frac 1{z-H_\pp} \right)^{-1} \left( \id - \frac Q{z-H_\pp} V_{\varphi-\pp} \right)^{-1}   \right| \psi^{\rm P} \right\rangle \frac{dz}{2\pi i} \,. \label{eq:ci2}
\end{align}

For simplicity, let us introduce the notation
\begin{equation}
A = V_{\varphi-\pp} \frac 1{z-H_\pp}  \quad , \quad B =  \frac Q{z-H_\pp} V_{\varphi-\pp} \,.
\end{equation}
Because of \eqref{3.12}  these operators are smaller than $1$ in norm, uniformly in $z\in C$. 
We shall use the identity 
\begin{align}\nonumber
\frac 1{\id-A} \frac 1{\id-B} & = \id + A  + A (A+B) +  \frac{B}{\id-B} \\ & \quad + \frac {A^3}{\id-A} +  \frac {A^2}{\id-A}  B +  \frac A{\id-A} \frac {B^2}{\id-B}\,.
\end{align}
We insert the various terms into \eqref{eq:ci2} and do the contour integration. The term $\id$ then yields $e(\pp)$. The term $A$ yields
\begin{equation}
 \langle \psi^{\rm P} | V_{\varphi - \pp} | \psi^{\rm P} \rangle = 2 \int_{\Omega} \pp(x) \left(\pp(x)  - \varphi(x) \right) dx
\end{equation}
using \eqref{eq:rvp}. A standard calculation shows that the term $A(A+B)$ leads to 
\begin{equation}
\langle \psi^{\rm P}  |  V_{\varphi- \pp}  \frac Q{ e(\pp) - H_{\pp}}  V_{\varphi - \pp}  | \psi^{\rm P}\rangle = - \langle \varphi- \pp |   K  | \varphi - \pp \rangle \,.
\end{equation}
Furthermore, since $Q| \psi^{\rm P}\rangle  = 0$, the term $B(\id-B)^{-1}$ yields zero. We conclude that 
\begin{align}\nonumber
 &\F^{\rm P}(\varphi)  - \F^{\rm P}(\pp)  -    \langle \varphi- \pp |  \id - K  | \varphi - \pp \rangle 
\\ & =   \int_C \frac z{z- e(\pp)} \left\langle \psi^{\rm P} \left|  \frac {A^3}{\id-A} +  A \left(  \frac {A}{\id-A}  +  \frac 1{\id-A} \frac {B}{\id-B}\right) B 
 \right| \psi^{\rm P} \right\rangle \frac{dz}{2\pi i}\,. \label{3.18}
\end{align}

To bound the first term on the right side of \eqref{3.18}, note that 
\begin{equation}\label{et1}
\left\langle \psi^{\rm P} \left|  \frac {A^3}{\id-A}  \right| \psi^{\rm P} \right\rangle = \frac 1{z-e(\pp)} \left\langle \psi^{\rm P} \left|  V_{\varphi-\pp} \frac 1{z-H_\pp} \frac {A}{\id-A} V_{\varphi-\pp}  \right| \psi^{\rm P} \right\rangle\,.
\end{equation}
We claim that 
\begin{equation}\label{cl1}
\sup_{z\in C} \left\| (-\Delta_\Omega)^{1/2}  \frac 1{z-H_\pp} \frac {A}{\id-A} (-\Delta_\Omega)^{1/2}  \right\| 
\lesssim \eps
\end{equation}
which implies that \eqref{et1} is bounded, in absolute value,  as 
\begin{equation}\label{eq:des}
| \eqref{et1} | \lesssim  \eps 
\, \left\langle \psi^{\rm P} \left|  V_{\varphi-\pp} (-\Delta_\Omega)^{-1}  V_{\varphi-\pp}  \right| \psi^{\rm P} \right\rangle =  \eps\, \langle \varphi- \pp |   L  | \varphi - \pp \rangle \,,
\end{equation}
as desired.  To prove \eqref{cl1} we use the fact that$\| (-\Delta_\Omega)^{1/2} (z-H_\pp)^{-1} (-\Delta_\Omega)^{1/2} \| $  is uniformly bounded   to reduce the problem to showing  $\| (-\Delta_\Omega)^{-1/2} A (\id - A)^{-1} (-\Delta_\Omega)^{1/2} \| \lesssim  \eps$. Since $S^{-1} A (\id-A)^{-1} S = S^{-1} A S (\id - S^{-1} A S)^{-1}$ with $S= (-\Delta_\Omega)^{1/2}$, it suffices to show that $\| (-\Delta_\Omega)^{-1/2} A  (-\Delta_\Omega)^{1/2} \| \lesssim  \eps$, which follows from 
  $\| (-\Delta_\Omega)^{-1/2} V_\varphi  (-\Delta_\Omega)^{-1/2} \| \leq  \|  V_\varphi  (-\Delta_\Omega)^{-1} \| $ and Lemma~\ref{lem:rb}.

For the last term in \eqref{3.18}, we simply bound 
\begin{align}\nonumber
& \left| 
\left\langle \psi^{\rm P} \left|  A \left(  \frac {A}{\id-A}  +  \frac 1{\id-A} \frac {B}{\id-B}\right) B 
 \right| \psi^{\rm P} \right\rangle \right| \\ & \leq \left\| \frac {A}{\id-A}  +  \frac 1{\id-A} \frac {B}{\id-B}  \right\|  \langle \psi^{\rm P} | A A^\dagger | \psi^{\rm P} \rangle ^{1/2} \langle \psi^{\rm P} | B^\dagger B | \psi^{\rm P} \rangle ^{1/2} \,.
\end{align}
The same bounds as above easily lead to the conclusion that also this term is bounded by the right side of \eqref{eq:des}. This concludes the proof of  Proposition \ref{hessian}.
\end{proof}


\subsection{A Uniform Quadratic Lower Bound}\label{ss:hessianglobal}

Inequality \eqref{lbh} gives a bound on $\F^{\rm P}$ for $\varphi$ near the minimizer $\pp$. We shall also need the following rougher global bound. 

\begin{proposition}\label{hessianglobal}
There is a constant $\kappa'>0$ such that for all $\varphi\in L_\R^2(\Omega)$,
\begin{equation}\label{lb2}
\F^{\rm P}(\varphi) \geq e^{\rm P} + \left\langle \varphi-\pp  \left| \id - \left(\id+\kappa' (-\Delta_\Omega)^{1/2} \right)^{-1} \right|  \varphi - \pp \right\rangle\,.
\end{equation}
\end{proposition}

We start with the following lemma.

\begin{lemma}\label{hessiangloballemma}
For $\psi \in H^1_0(\Omega)$ with $\|\psi\|_2=1$, 
\begin{equation}
\left\langle |\psi|^2 -|\psi^{\rm P}|^2 \right| (-\Delta_\Omega)^{-1/2} \left| |\psi|^2 -|\psi^{\rm P}|^2 \right \rangle \leq \frac 8{\pi^2}  \int_\Omega \left| \nabla \left( |\psi| - |\psi^{\rm P}| \right) \right|^2 \,.
\end{equation}
\end{lemma}

\begin{proof}
With $f(x) = |\psi(x)| + |\psi^{\rm P}(x)|$ and $g(x) = |\psi(x)| - |\psi^{\rm P}(x)|$, the Schwarz inequality and the symmetry and positivity of the integral kernel of $(-\Delta_\Omega)^{-1/2}$ imply that 
\begin{align}\nonumber
&\left\langle |\psi|^2 -|\psi^{\rm P}|^2 \right| (-\Delta_\Omega)^{-1/2} \left|  |\psi|^2 -|\psi^{\rm P}|^2  \right \rangle \\ \nonumber 
&= \int_\Omega \int_\Omega f(x) g(x) (-\Delta_\Omega)^{-1/2}(x,y) f(y) g(y) \, dx\, dy  \\ &\leq  \int_\Omega \int_\Omega f(x)^2 (-\Delta_\Omega)^{-1/2}(x,y)  g(y)^2 \, dx\, dy \,.
\end{align}
For fixed $x$, we can use the Hardy inequality and the fact that $(-\Delta_\Omega)^{-1/2}(x,y)\leq (-\Delta_{\R^3})^{-1/2}(x,y)= (2\pi^2)^{-1} |x-y|^{-2}$ from \eqref{eq:b2}
to obtain the bound
\begin{equation}
  \int_\Omega  (-\Delta_\Omega)^{-1/2}(x,y)  g(y)^2 \, dy \leq \frac2{\pi^2}  \int_\Omega \left| \nabla g \right|^2 \,.
\end{equation}
Since $\int_\Omega f^2 \leq 4$, the result follows.   
\end{proof}

\begin{proof}[Proof of Proposition \ref{hessianglobal}]
From our assumption \eqref{ass:hp} on the Hessian of the Pekar functional $\E^{\rm P}$ and Lemma~\ref{hessiangloballemma}, it follows that 
\begin{equation}
\E^{\rm P}(\psi) \geq \E^{\rm P}(|\psi|)\geq \E^{\rm P}(\psi^{\rm P}) +  \kappa' \left\langle |\psi|^2 -|\psi^{\rm P}|^2 \right| (-\Delta_\Omega)^{-1/2} \left|  |\psi|^2 -|\psi^{\rm P}|^2   \right \rangle
\end{equation}
for  $\kappa' = \kappa \pi^2/8$. In particular,
\begin{align}\nonumber
\E(\psi,\varphi) & = \E^{\rm P}(\psi) + \left\| \varphi - (-\Delta_\Omega)^{-1/2} |\psi|^2 \right\|_2^2 \\
& \geq e^{\rm P} +  \kappa' \left\langle |\psi|^2 -|\psi^{\rm P}|^2 \right| (-\Delta_\Omega)^{-1/2} \left|  |\psi|^2 -|\psi^{\rm P}|^2   \right \rangle +  \left\| \varphi - (-\Delta_\Omega)^{-1/2} |\psi|^2 \right\|_2^2\,.
\end{align}
Minimizing with respect to $\psi$ and using \eqref{eq:rvp} leads to the desired lower bound. 
\end{proof}


\section{Proof of Theorem~\ref{main:thm}: Upper bound}\label{sec:upper}

In this section we construct a trial state to derive an upper bound on the polaron ground state energy. 
Our trial state $\Psi$ will depend only on finitely many phonon variables. More precisely, for $\Pi$ a finite rank projection on $L^2_\R(\Omega)$, we can write the Fock space $\F(L^2(\Omega))$ as a tensor product $\F(\Pi L^2(\Omega)) \otimes \F((\id-\Pi)L^2(\Omega))$, and our trial state corresponds to the vacuum vector in the second factor $\F((\id-\Pi)L^2(\Omega))$. The first factor $\F(\Pi L^2(\Omega))$ can naturally be identified with $L^2(\R^N)$ corresponding to $N$ simple harmonic oscillators, where $N = \dim \ran \Pi$. 

We find it convenient to identify a point $(\lambda_1,\dots,\lambda_N) \in \R^N$ with a function $\varphi \in \ran \Pi$ via the identification 
\begin{equation}\label{eq:ident_up}
 \varphi = \Pi \varphi = \sum_{n=1}^N \lambda_n \varphi_n 
\end{equation}
for some orthonormal basis $\{\varphi_n\}$ of $\ran \Pi$. With this identification, we can think of a wave function $\Psi \in L^2(\Omega) \otimes L^2(\R^N)$ as a function $\Psi(x,\varphi)$ with $x\in\Omega, \varphi\in \ran\Pi$. 
The function we choose is as follows:
\begin{equation}
\Psi(x,\varphi) = e^{ - \alpha^2  \left\langle \varphi- \pp\left|  (\id - \Pi K \Pi )^{1/2} \right| \varphi- \pp  \right\rangle} 
\chi\left( \epsilon^{-1} \| (-\Delta_\Omega)^{-1/2} ( \varphi - \pp) \|_2 \right)  \psi_{\varphi}(x)
\end{equation}
where
\begin{itemize}
\item $\epsilon>0$ is a small parameter that will be chosen to go to zero as $\alpha\to\infty$.
\item $0\leq \chi\leq 1$ is a smooth cut-off function with $\chi(t) = 1$ for $t\leq 1/2$ and $\chi(t)=0$ for $t\geq 1$ 
\item $\Pi$ is a finite rank projection on $L^2_\R(\Omega)$, with range containing $\pp$.
\item $\psi_\varphi$ is the unique non-negative, normalized ground state of $H_\varphi=-\Delta_\Omega+V_\varphi$ 
\item $K = \id-H^{\rm P}$, explicitly given in \eqref{def:K}.
\end{itemize}

On states of the type described above (corresponding to the vacuum for all modes outside the range of $\Pi$),   the Hamiltonian \eqref{def:ham} simply acts as $H_\varphi + \N$, with
\begin{equation}
 \N  =   \sum_{n=1}^N \left(  - \frac {1}{4 \alpha^4} \partial_{\lambda_n}^2  + \lambda_n^2 - \frac {1}{2 \alpha^2} \right)\,. 
\end{equation}
Using the eigenvalue equation $H_\varphi \psi_\varphi = e(\varphi)\psi_\varphi$, the energy of our trial state $\Psi$ is thus given as 
\begin{equation}
\left\langle \Psi\left| \Hh \right| \Psi \right\rangle = \left\langle \Psi \left|  e(\varphi) + \N \right| \Psi \right\rangle\,.
\end{equation}
Since $\Psi$ is supported on the set $\{  \| (-\Delta_\Omega)^{-1/2} ( \varphi - \pp) \|_2\leq \epsilon\}$, we can use Proposition~\ref{hessian}
for an upper bound on $e(\varphi)$. This leads to 
\begin{equation} \label{frrso}
\left\langle \Psi \left| \Hh \right| \Psi \right\rangle  \leq e^{\rm P}   \langle \Psi |  \Psi \rangle + 
\left\langle \Psi \left|  \N  - \|\varphi\|_2^2 +  \langle \varphi- \pp |  \id - K + \epsilon C L | \varphi - \pp \rangle
   \right| \Psi \right\rangle
\end{equation}
for a suitable constant $C>0$. 

 Utilizing the fact that the Gaussian factor in $\Psi$ satisfies 
\begin{align}\nonumber
&\left( - \frac 1{4\alpha^{4}} \sum_{n=1}^N  \partial_{\lambda_n}^2 +  \langle \varphi- \pp |  \id - K|  \varphi - \pp \rangle \right)  e^{-{ \alpha^2}  \left\langle \varphi- \pp\left| (\id - \Pi K \Pi )^{1/2} \right| \varphi- \pp  \right\rangle} \\ & =  \frac 1{2 \alpha^{2}} \tr  (\id -  \Pi K \Pi )^{1/2}\, e^{-{ \alpha^2}  \left\langle \varphi- \pp\left| (\id - \Pi K \Pi )^{1/2} \right| \varphi- \pp  \right\rangle} 
\end{align}
we can integrate by parts and rewrite the right side of \eqref{frrso} as 
\begin{equation}
\left(  e^{\rm P} - \frac 1 {2 \alpha^{2}} \tr   \left[ \id - (\id -  \Pi K \Pi )^{1/2} \right]  \right)  \langle \Psi |  \Psi \rangle + A + B
\end{equation}
with 
\begin{equation}
A =  \epsilon  C \left\langle \Psi \left|   \left\langle \varphi- \pp \left|   L \right| \varphi - \pp \right\rangle \right| \Psi \right\rangle 
\end{equation}
and
\begin{align}\nonumber
B & = \frac 1{4 \alpha^{4}}  \sum_{n=1}^N \int_{\Omega} dx \int_{\R^N} \prod_{m=1}^N d \lambda_m \, e^{- 2{\alpha^2} \left\langle \varphi- \pp\left| (\id -  \Pi K\Pi  )^{1/2} \right| \varphi- \pp \right \rangle}  \\ & \qquad\qquad \qquad\qquad\qquad\times \left|    \partial_{\lambda_n} \left( \chi\left( \epsilon^{-1} \| (-\Delta_\Omega)^{-1/2}( \varphi - \pp) \|_2 \right)  \psi_{\varphi}(x) \right) \right|^2\,.
\end{align}

We claim that $L$ is bounded by $(-\Delta_\Omega)^{-1}$. This follows immediately from the boundedness of $\psi^{\rm P}$ shown in Lemma \ref{psibounded}. Alternatively, one can use that $\psi (-\Delta_\Omega)^{-1} \psi$ is a bounded operator for $\psi \in L^3(\Omega)$ by Sobolev's inequality.  Hence we  can use the rough bound
\begin{equation}
A \lesssim  \epsilon^3   \langle \Psi|  \Psi \rangle \,.
\end{equation}
Moreover, by a simple Cauchy--Schwarz inequality, $B \leq 2 (B_1+B_2)$ with
\begin{align}\nonumber
B_1 & = \frac 1 {4\alpha^{4}} \int_{\R^N} \prod_{m=1}^N d\lambda_m  \, e^{- 2{\alpha^2}  \left\langle  \varphi- \pp\left| (\id -  \Pi K \Pi )^{1/2}\right|  \varphi- \pp \right \rangle} \\ & \qquad \qquad\qquad\qquad  \times \chi\left( \epsilon^{-1} \| (-\Delta_\Omega)^{-1/2} (  \varphi - \pp )\|_2 \right)^2 \sum_{n=1}^N  \left\| \partial_{\lambda_n}  \psi_{\varphi} \right\|_2^2 
\end{align}
and
\begin{align}\nonumber 
B_2 & = \frac 1{4\alpha^{4}} \int_{\R^N} \prod_{m=1}^N d \lambda_m \,  e^{- 2 {\alpha^2}  \left\langle  \varphi- \pp\left| (\id -  \Pi K \Pi )^{1/2} \right|  \varphi- \pp \right \rangle}  \\
& \qquad \qquad \qquad\qquad\times \sum_{n=1}^N \left| \partial_{\lambda_n} \chi\left( \epsilon^{-1} \| (-\Delta_\Omega)^{-1/2} ( \varphi - \pp) \|_2 \right)   \right|^2\,.
\end{align}

To bound $B_1$, we  use standard first order perturbation theory for eigenvectors to compute
\begin{equation}
\partial_{\lambda_n} \psi_{\Pi\varphi} = - \frac {Q_{\varphi}}{H_{\varphi} - e(\varphi)}   V_{\varphi_n}  \psi_{\varphi}
\end{equation}
where $Q_{\varphi}= \id - |\psi_{\varphi}\rangle\langle\psi_{\varphi}|$.  In particular, 
\begin{equation}\label{4.14}
\sum_{n=1}^N \left\|  \partial_{\lambda_n} \psi_{\varphi} \right\|_2^2 = 4 \Tr \Pi  (-\Delta_\Omega)^{-1/2} \psi_{\varphi} \left(  \frac {Q_{\varphi}}{H_{\varphi} - e(\varphi)}   \right)^2 \psi_{\varphi} (-\Delta_\Omega)^{-1/2} \Pi
\end{equation}
where we again interpret $\psi_{\varphi}$ as multiplication operator on the right side. It is not difficult to see that $(-\Delta_\Omega)^{1/2} \frac {Q_{\varphi}}{H_{\varphi} - e(\varphi)} (-\Delta_\Omega)^{1/2} $ is uniformly bounded on the support of $\chi$ (compare with the proof of Proposition~\ref{hessian}). Using this fact and 
\eqref{argu:hs}, we 
 see that  \eqref{4.14}  is uniformly bounded, independently of $N$. Hence $B_1 \lesssim   \alpha^{-4} \langle \Psi | \Psi \rangle$.

For $B_2$, we have 
\begin{align}\nonumber 
B_2 & \lesssim \frac 1{\alpha^{4} \epsilon^{2}}   \int_{\R^N} \prod_{m=1}^N d \lambda_m  \,  e^{- 2 \alpha^2  \left\langle \varphi- \pp\left| (\id -  \Pi K\Pi )^{1/2}\right| \varphi- \pp  \right\rangle} \\ & =  \frac 1{\alpha^{4} \epsilon^{2}}   \left( \alpha \sqrt{  2/ \pi} \right)^{-N}\det( \id-\Pi K \Pi )^{-1/4}
\end{align}
where we have used the fact that $\varphi^{\rm P}$ is in the range of $\Pi$. 
We have to compare this with the norm of $\Psi$, which is bounded from below by 
\begin{align}\nonumber
\langle \Psi| \Psi\rangle  & \geq \int_{\R^N\setminus S_\epsilon} \prod_{m=1}^N  d \lambda_m  \,  e^{- \alpha^2  \left\langle \varphi- \pp\left| (\id -  \Pi K \Pi )^{1/2} \right|  \varphi- \pp \right \rangle}   \\& = \left( \alpha \sqrt{  2/ \pi} \right)^{-N}  \det( \id-\Pi K\Pi )^{-1/4} -  \int_{S_\epsilon} \prod_{m=1}^N d \lambda_m  \,  e^{- 2\alpha^2  \left\langle \varphi- \pp\left| (\id -  \Pi K\Pi )^{1/2}\right|  \varphi- \pp \right \rangle}  
\end{align}
where 
\begin{equation}
S_\epsilon = \left\{ \vec \lambda \in \R^N\, : \,   \| (-\Delta_\Omega)^{-1/2} ( \varphi - \pp)\|_2 \geq  \epsilon/2  \right\}\,.
\end{equation}
Since $\|K\|<1$ by assumption,  $(-\Delta_\Omega)^{-1} \leq \nu (1-\Pi K\Pi)^{1/2}$ for some constant $\nu>0$ independent of $N$. Hence  
we can bound the characteristic function of $S_\epsilon$ from above by $\exp(-\frac1{4\nu} \alpha^2 \epsilon^2) \times \exp( \alpha^2   \left\langle \varphi- \pp\left| (\id -  \Pi K \Pi )^{1/2} \right| \varphi- \pp \right \rangle )$. Therefore,
\begin{equation}
\langle \Psi| \Psi\rangle  \geq \left( \alpha \sqrt{  2/ \pi} \right)^{-N}  \det( \id-\Pi K \Pi )^{-1/4} \left( 1- 2^{N/2} e^{-\frac1{4\nu} \alpha^2 \epsilon^2} \right)\,.
\end{equation}
In particular, as long as $\alpha\epsilon \geq \const \sqrt N$ with a sufficiently large constant, we have  $ \langle \Psi| \Psi\rangle  \gtrsim \left( \alpha\sqrt{2/\pi} \right)^{-N}    \det( \id-\Pi K \Pi )^{-1/4}$, and hence  
\begin{equation}
B_2 \lesssim  \alpha^{-4} \epsilon^{-2}   \langle \Psi| \Psi\rangle\,.
\end{equation}

In summary, we have shown that 
\begin{equation}
\frac {\langle \Psi | \Hh | \Psi \rangle} { \langle \Psi | \Psi\rangle}  \leq  e^{\rm P} -  \frac 1{2\alpha^{2}} \tr   \left[ \id - (\id -  \Pi K \Pi )^{1/2} \right]  + \const \left( \epsilon^3 + \alpha^{-4}\epsilon^{-2} \right)
\end{equation}
as long as $\alpha\epsilon \geq \const \sqrt N$ and $\epsilon$ is small enough. We shall choose $\Pi$ to be the projection onto the span of $g_1,\dots,g_{N-1},\varphi^{\rm P}$, where we denote by $\{g_j\}_j$ an orthonormal basis of eigenfunctions of $K$, ordered in a way that the corresponding eigenvalues $k_j = \langle g_j| K g_j\rangle$ form a decreasing sequence$^1$.\footnote{$^1$In case $\varphi^{\rm P}$ is in the span of $\{g_j\}_{j=1}^{N-1}$, we take $\Pi$ to be the projection onto the span of $\{g_j\}_{j=1}^N$ instead.} 
Then
\begin{equation}
\tr   \left[ \id - (\id -  \Pi K \Pi )^{1/2} \right] \geq \sum_{j=1}^{N-1} \left( 1 - \left( 1 - k_j\right)^{1/2} \right)
\end{equation}
and hence 
\begin{equation}
\tr   \left[ \id - (\id -  \Pi K \Pi )^{1/2} \right] \geq \tr   \left[ \id - (\id -   K  )^{1/2} \right]  - \sum_{j=N}^{\infty} \left( 1 - \left( 1 - k_j\right)^{1/2} \right)\,.
\end{equation}
Since $K \lesssim [ (-\Delta_\Omega)^{-1/2} \psi^{\rm P} (-\Delta_\Omega)^{-1/2}]^2$ and $\psi^{\rm P}$ is bounded by Lemma~\ref{psibounded}, we have $k_j \leq \const e_j^{-2}$, where $e_j$ denotes the (ordered) eigenvalues of $-\Delta_\Omega$.  Since $\Omega$ is assumed to be a smooth and bounded domain, we have the Weyl asymptotics $e_j \sim j^{2/3}$ for $j\gg 1$ (see, e.g., \cite[Sec.~XIII.15]{RS4}), which implies that 
\begin{equation}
\sum_{j=N}^{\infty} \left( 1 - \left( 1 - k_j\right)^{1/2} \right) \lesssim N^{-1/3}\,.
\end{equation}

In order to minimize the error term, we shall choose  $\epsilon \sim \alpha^{-8/11}$ and $N\sim \alpha^2 \epsilon^2 \sim \alpha^{6/11}$, which leads to the bound 
\begin{equation}
\frac {\langle \Psi | \Hh | \Psi \rangle} { \langle \Psi | \Psi\rangle}  \leq  e^{\rm P} - \frac 1{2 \alpha^{2}} \tr   \left[ \id - (\id - K  )^{1/2} \right]  + \const \alpha^{-24/11} 
\end{equation}
for large enough $\alpha$. This concludes the proof of the upper bound in Theorem~\ref{main:thm}. 
\hfill\qed


\section{Multiple Lieb--Yamazaki Bound}\label{sec:LY}

In \cite{LY} Lieb and Yamazaki used the fact that the interaction between the particle and the field can be written as a commutator,  together with a Cauchy--Schwarz inequality, to get a uniform lower bound on the ground state energy of $\Hh$ (for $\Omega=\R^3$) for large $\alpha$. Their method shows that the introduction of an ultraviolet cutoff $\Lambda$ in the interaction affects the ground state energy at most by $O(\Lambda^{-1/2})$. We shall apply their method multiple (in fact, three) times, which will allow us to conclude that the effect of the cutoff is  at most $O(\Lambda^{-5/2})$ (up to logarithmic corrections). It will be essential to use the Gross transformation explained in the next section, however, since we need relative operator boundedness of the kinetic energy with respect to the full Hamiltonian, which only holds for the transformed kinetic energy, as we shall see.

Before stating the main result of this section,  we shall prove the following useful lemma. Its proof proceeds similarly to the one of Lemma~10 in \cite{frankschlein}. For its statement, we introduce the Coulomb norm
\begin{equation}
\| f\|_{\rm C} =  \left( \frac 1{ 4\pi} \int_{\R^6} \frac {\overline{f(x)} f(y)} {|x-y|} \, dx\, dy \right)^{1/2}\,.
\end{equation}
By the Hardy--Littlewood--Sobolev inequality (see, e.g., \cite[Thm. 4.3]{LiLo}), this norm is dominated by the $L^{6/5}(\R^3)$-norm.

Let us introduce the notation $p = -i\nabla_x = (p_1,p_2,p_3)$ for the momentum operator. We shall also use $p^2$ for the Dirichlet Laplacian $-\Delta_\Omega$ on $\Omega$.

\begin{lemma}\label{lem:FS}
Consider a function $h_x(\,\cdot\,)$ such that $ k(x) = \sup_{y\in\R^3} |h_{x+y}(y)|$ has finite Coulomb norm. Then
\begin{equation}\label{eq:fschl}
 a^\dagger(h_x)   a(h_x)  \leq   \| k \|_{\rm C}^2 \, p^2 \N 
\end{equation}
holds on $L^2(\Omega)\otimes \mathcal{F}$. 
\end{lemma}

Note that the bound holds trivially with the right side replaced by $\|h_x\|_2^2 \N$. The point of Lemma~\ref{lem:FS} is that functions that are more singular (in the $x-y$ variable) can be handled, at the expense of the kinetic energy term $p^2$. 

\begin{proof}
For convenience of notation, let $\Psi$ be a one-phonon vector; the general case works in the same way. We need to bound
\begin{equation}
\int_{\Omega} \left| \int_{\Omega} \Psi(x,y) h_x(y) dy \right|^2 dx \leq \int_{\Omega} \left| \int_{\Omega} | \Psi(x,y)|  k(x-y ) dy \right|^2 dx\,.
\end{equation}
With $\Phi(p,q)$ denoting the Fourier transform of $|\Psi(x,y)|$ (regarded as a function on $\R^3\times \R^3$), we have
\begin{align}\nonumber
 & \int_{\Omega} \left| \int_{\Omega} | \Psi(x,y)|  k(x-y)\, dy \right|^2  dx
 \\ \nonumber &  =  \int_{\R^3} \left|   \int_{\R^3} \Phi(p-q  ,q ) \hat k(q)  \, dq \right|^2 dp 
 \\ \nonumber & \leq \int_{\R^3} \left( \int_{\R^3} |\Phi(p-q ,q)|^2 (p-q)^2 \, dq \right) \left( \int_{\R^3} | \hat k(q)|^2  |p-q|^{-2} \, dq \right) dp 
 \\ & \leq \sup_p  \left( \int_{\R^3} | \hat k(q)|^2  |p-q|^{-2} \, dq \right)  \,  \int_{\R^3}  \int_{\R^3} |\Phi(p ,q)|^2 p^2 \, dq \, dp \,. 
\end{align}
The last factor is smaller than  $\| \sqrt{\N}\sqrt{p^2} \Psi \|^2$ (by the diamagnetic inequality). By writing the integral in $x$-space, one easily checks that 
\begin{equation}
\sup_p  \int_{\R^3} | \hat k(q)|^2  |p-q|^{-2} \, dq \leq \|k\|_{\rm C}^2\,,
\end{equation}
hence our claim \eqref{eq:fschl}  is  proven.
\end{proof}

The main result of this section is the following.

\begin{lemma}\label{lem:3LY}
Assume that $w_x(\,\cdot\,)$ is such that 
\begin{equation}\label{def:A1}
A_1 \coloneqq \max_{j,k,l} \sup_{x\in\Omega} \| p_j p_k p_l |p|^{-6} w_x\|_2 < \infty
\end{equation}
\begin{equation}
A_2 \coloneqq   \max_{j,k} \sup_{x\in\Omega} \| p_j p_k |p|^{-4} w_x\|_2  < \infty
\end{equation}
and 
\begin{equation}\label{def:A3}
A_3\coloneqq \max_{j,k}  \|u_{jk}\|_{\rm C} < \infty
\end{equation}
where  $u_{jk}(x) = \sup_{y\in \R^3} | p_j p_k |p|^{-4} w_{x+y}(y)|$. 
Then 
 \begin{align}\nonumber
a(w_x) + a^\dagger( w_x) & \leq 12  A_1  \left(    |p|^4  + 3 p^2  \left(  \N +  1/(2\alpha^{2}) \right)  \right)  \\ & \quad + 6 \left( \alpha^{-1}  A_2  + A_3 \right)  \left( |p|^4 + p^2 \N +\frac 12  \right)  \label{lemC}
\end{align}
holds on $L^2(\Omega)\otimes \mathcal{F}$. 
\end{lemma}

\begin{proof}
For any $w_x$, we have
\begin{equation}
\sum_j\, [ p_j, a(p_j |p|^{-2} w_x)] = - a(w_x) \,.
\end{equation}
Applying this three times, we also get
\begin{equation}
\sum_{j,k,l}\, [ p_j , [ p_k, [p_l , a(p_j p_k p_l |p|^{-6} w_x) ] ]]  = - a (  w_x)    \,.
\end{equation}
In particular, we conclude that 
\begin{equation}\label{el1}
a(w_x) + a^\dagger(w_x) =  \sum_{j,k,l}\, [ p_j , [ p_k, [p_l , a^\dagger(p_j p_k p_l |p|^{-6} w_x) - a(p_j p_k p_l |p|^{-6}w_x) ] ]]\,.
\end{equation}
We introduce the notation $B_{jkl} = a^\dagger(p_j p_k p_l |p|^{-6} w_x) - a(p_j p_k p_l |p|^{-6} w_x)$,  and rewrite the triple commutator as
\begin{align}\nonumber
\sum_{j,k,l}\, [ p_j , [ p_k, [p_l , B_{jkl}  ] ]] & = \sum_{j,k,l}\, \left(  p_j p_k [p_l, B_{jkl}] + [B_{jkl}^\dagger, p_l] p_j p_k  \right) \\ &  \quad  - 2 \sum_{j,k,l}\, \left(  p_j p_k  B_{jkl} p_l + p_l B_{jkl}^\dagger p_j p_k  \right)\label{el2}
\end{align}
using the invariance of $B_{jkl}$ under exchange of indices. 

The Cauchy--Schwarz inequality implies that 
\begin{equation}
 - p_j p_k  B_{jkl} p_l -  p_l B_{jkl}^\dagger p_j p_k \leq  \lambda p_j^2 p_k^2 + \lambda^{-1} p_l B_{jkl}^\dagger B_{jkl} p_l
\end{equation}
for any $\lambda>0$. 
Moreover,
\begin{equation}
B_{jkl}^\dagger B_{jkl}  \leq  \left( 4 \N +2  \alpha^{-2} \right) \| p_j p_k p_l |p|^{-6} w_x\|_2^2 \leq A_1^2 \left( 4 \N +2  \alpha^{-2} \right)\,.
\end{equation}
In particular, by choosing $\lambda= 2 A_1$ and summing over $j,k,l$, we obtain the bound 
\begin{equation}\label{el3}
  - 2 \sum_{j,k,l}\, \left(  p_j p_k  B_{jkl} p_l + p_l B_{jkl}^\dagger p_j p_k  \right) \leq 
 12  A_1 \left(    |p|^4  + 3 p^2  \left(  \N +  1/(2 \alpha^{2}) \right)  \right)\,.
\end{equation}

We also have
\begin{equation}
C_{jk} = \sum_l  [p_l, B_{jkl}]  = a^\dagger(p_j p_k |p|^{-4} w_x) + a(p_j p_k |p|^{-4} w_x)
\end{equation}
and 
\begin{equation}
p_j p_k C_{jk}  + C_{jk} p_j p_k   \leq  \lambda p_j^2 p_k^2 + \lambda^{-1} C_{jk}^2
\end{equation}
for any $\lambda>0$. 
Furthermore, we can bound
\begin{equation}
C_{jk}^2  \leq  4 a^\dagger(p_j p_k |p|^{-4} w_x)  a(p_j p_k |p|^{-4} w_x) + \frac 2{\alpha^2} \| p_j p_k |p|^{-4} w_x\|_2^2\,.
\end{equation}
By Lemma~\ref{lem:FS}, the first term on the right side is bounded by  $ 4 \| u_{jk} \|_{\rm C}^2\, p^2  \N$, and hence $C_{jk}^2 \leq 4 A_3^2 p^2 \N+ 2 A_2^2 \alpha^{-2}$.  The choice $\lambda = 6  ( A_3 + \alpha^{-1} A_2)$ then leads to the bound
\begin{equation}
 \sum_{j,k}\, \left(  p_j p_k C_{jk} + C_{jk} p_j p_k  \right) \leq 
6  \left( A_3 + \alpha^{-1} A_2\right) \left( |p|^4 + p^2 \N +\frac 12  \right) \,.
\end{equation}
In combination with \eqref{el1}, \eqref{el2} and \eqref{el3}, this concludes the proof of the lemma.
\end{proof}

In the following, we shall apply this bound to the large momentum part of the interaction, in order to quantify the effect of an ultraviolet cutoff on the ground state energy. Because the Coulomb norm in \eqref{def:A3} estimates the off-diagonal decay, we cannot use a sharp cutoff, however, and need to work with a smooth one instead. In fact, we shall apply Lemma~\ref{lem:3LY}  with 
\begin{equation}\label{def:w}
w_x(y) = z(-\Delta_\Omega)(x,y) \quad \text{for} \quad z(t) = t^{-1/2} \left( 1 - e^{-t/\Lambda^2} \right)^2
\end{equation}
for some $\Lambda>0$. The function $z$ is non-negative, and behaves like $t^{3/2} \Lambda^{-4} $ for $t\ll \Lambda^2$. Moreover, $z(t) - t^{-1/2}$ falls off like $t^{-1/2} e^{-t/\Lambda^2}$ for $t \gg \Lambda^2$. 

We shall show in Appendix~\ref{sec:appb} that the various norms appearing in \eqref{def:A1}--\eqref{def:A3} can be bounded, up to a multiplicative constant, by the equivalent expressions for $\Omega=\R^3$, which can easily be estimated using Fourier transforms. We have 
\begin{equation}
 \| p_j p_k |p|^{-4} w_x\|_2^2 = 
 \sum_n e_n^{-5} \left( 1 - e^{-e_n/\Lambda^2} \right)^4 |\partial_j \partial_k \varphi_n(x)|^2
 \end{equation}
where $e_n$ and $\varphi_n$ denote the eigenvalues and eigenfunctions of $-\Delta_\Omega$. In particular, from \eqref{eq:b13} we deduce that
\begin{equation}\label{bound:a1}
 \sup_{x\in \Omega}  \max_{j,k}  \| p_j p_k |p|^{-4} w_x\|_2 \lesssim \left( \int_{\R^3} |k|^{-6} \left( 1 - e^{-k^2/\Lambda^2} \right)^4 dk \right)^{1/2} = \const \Lambda^{-3/2}\,.
\end{equation}
In the same way, we obtain the bound
\begin{equation}\label{claim1}
 \sup_{x\in \Omega}  \max_{j,k,l} \| p_j p_k p_l |p|^{-6} w_x\|_2 \lesssim \Lambda^{-5/2}\,.
\end{equation}
Moreover, in Sect.~\ref{ss:b3} we shall show that
\begin{equation}\label{bound:a3}
\max_{j,k} \|u_{jk}\|_{\rm C}\lesssim \Lambda^{-5/2}\,.
\end{equation}
We collect these results in the following corollary.

\begin{corollary} \label{cor:LY}
For $\Lambda>0$ let $w_x(\, \cdot\,)$ be the function defined in (\ref{def:w}). 
Then
\begin{equation}
a(w_x) + a^\dagger(w_x) \lesssim    \left( p^2 + \N + 1  \right)^2 \left( \Lambda^{-5/2}   +  \alpha^{-1}  \Lambda^{-3/2}  \right) 
\end{equation}
for $\alpha\gtrsim 1$. 
\end{corollary}


\section{Gross Transformation}\label{section:GT}

In this section we shall investigate the effect of a unitary Gross transformation \cite{grossT,nelson} on the Hamiltonian~\eqref{def:ham}. 
Let $\{f_x\}_{x\in\Omega}\subset L^2(\Omega)$ be a family of functions, parametrized by $x\in\Omega$, such that $\nabla_x f_x \in L^2(\Omega)$ for all $x\in\Omega$. We consider a unitary transformation in $L^2(\Omega)\otimes\mathcal F$ of the form 
\begin{equation}\label{def:U}
U = e^{ a(\alpha^2 f_x) - a^\dagger(\alpha^2 f_x)}\,.
\end{equation}
(This operator acts by `multiplication' with respect to the $x$ variable.) For $g\in L^2(\Omega)$ we have
\begin{equation}
U a(g) U^\dagger = a(g) + \langle g| f_x\rangle
\qquad\text{and}\qquad
U a^\dagger(g) U^\dagger = a^\dagger(g) + \langle f_x| g\rangle
\end{equation}
and hence
\begin{equation}
U \N U^\dagger = \N + a^\dagger(f_x) + a(f_x) + \|f_x\|_2^2\,.
\end{equation}
Moreover, for $p=-i\nabla_x$,
\begin{equation}
\label{eq:unabla}
U p U^\dagger = p +  \alpha^2 \left( a^\dagger( p f_x) + a(p f_x)  +  {\rm Re}\, \langle  f_x | p f_x \rangle \right)\,.
\end{equation}
We shall choose $f_x$ real-valued, hence the last term vanishes. Then
\begin{align}\nonumber
Up^2U^\dagger &= p^2 + \alpha^4 \left( a^\dagger( p f_x) + a(p f_x) \right)^2 \\ & \quad + 2 \alpha^2 p\cdot  a(pf_x) + 2 \alpha^2  a^\dagger(p f_x)\cdot p + \alpha^2 a ( p^2 f_x) + \alpha^2 a^\dagger(p^2 f_x)\,.
\end{align}

For the Hamiltonian \eqref{def:ham}, we thus have
\begin{align}\nonumber
U \Hh U^\dagger &= p^2 + \alpha^4 \left( a^\dagger( p f_x) + a(p f_x) \right)^2 + 2 \alpha^2 p\cdot a(p f_x) + 2 \alpha^2 a^\dagger(p f_x)\cdot p \\ & \quad 
 + a ( \alpha^2 p^2 f_x +f_x - v_x) + a^\dagger( \alpha^2 p^2 f_x + f_x - v_x)  
+ \N  + \|f_x\|_2^2- 2 {\rm Re}\, \langle v_x|f_x\rangle  \,. \label{uhu}
\end{align}
We shall choose $f_x$ such that
$
\alpha^2 p^2 f_x + f_x - v_x = g_x$, i.e., 
\begin{equation}
  f_{\,\cdot\,}(y)  = \left( -\alpha^2 \Delta_\Omega + 1\right)^{-1} \left(  g_{\, \cdot\, }(y)  + v_{\, \cdot\, } (y) \right) \quad \forall y\in\Omega 
\end{equation}
for some $g_x\in L^2(\Omega)$ with $\sup_{x\in\Omega} \|g_x\|_2 < \infty$. The choice $g_x\equiv 0$ would be possible, but it will be more convenient to choose 
\begin{equation}
g_x(y) = \xi(-\Delta_\Omega)(x,y) \quad \text{for}\quad \xi(t) = -t^{-1/2} \theta(K^2 - t)
\end{equation}
for some $K>0$, where 
\begin{equation}
\theta(t) = \left\{ \begin{array}{ll} 0 & \text{for $t<0$} \\ $1/2$ & \text{for $t=0$} \\ 1 & \text{for $t>0$.}\end{array} \right.
\end{equation}
Then
\begin{equation}
\|g_x\|_2^2 = \xi^2(-\Delta_\Omega)(x,x)
\end{equation}
and, since  $
\xi(t)^2 \leq  t^{-1} e^{1-t/K^2}
$, 
the fact that the heat kernel of $\Delta_\Omega$ is dominated by the one of $\Delta_{\R^3}$ implies as in \eqref{eq:b2} that 
\begin{equation}\label{bound:g}
\sup_x \|g_x\|_2^2  \leq  \frac 1{(2\pi)^3}  \int_{\R^3}   \frac {e^{1- k^2/K^2}} {k^2 }  dk =  \frac e{ 4 \pi^{3/2}} K\,.
\end{equation}

For the corresponding $f_x$, we  have 
\begin{equation}\label{def:f}
f_x(y) = \eta(-\Delta_\Omega)(x,y) \quad \text{for}\quad \eta(t) = -t^{-1/2} \frac{ \theta(t- K^2)}{\alpha^2 t + 1}\,.
\end{equation}
Using the fact that 
\begin{equation}\label{eq:613}
\eta(t)^2 \leq \alpha^{-4} t^{-3} \theta(t-K^2) \leq \alpha^{-4} \left( \frac 2 { t + K^2} \right)^3
\end{equation}
one obtains in a similar way as above 
\begin{equation}\label{bound:f}
\sup_x \| f_x \|_2^2 = \sup_x \eta^2(-\Delta_\Omega)(x,x) \leq \frac 1{\alpha^4 (2\pi)^3}  \int_{\R^3} \left(  \frac {2}{k^2 + K^2} \right)^3 dk  =  \frac 1{4\pi } \alpha^{-4} K^{-3}\,.
\end{equation}
Moreover,
\begin{equation}\label{bound:vf}
\sup_x |\langle v_x|f_x\rangle|  \leq \frac 1{\alpha^2 (2\pi)^3}  \int_{\R^3} \left(  \frac {2}{k^2 + K^2} \right)^2 dk   = \frac 1{2\pi}  \alpha^{-2} K^{-1}
\end{equation}
and, using \eqref{eq:613} and \eqref{eq:b13},
\begin{equation}\label{bound:pf}
\sup_x \| p f_x\|_2^2 \lesssim  \frac {1}{\alpha^4}  \int_{\R^3} k^2 \left(  \frac {2}{k^2 + K^2} \right)^3 dk  = 6 \pi^2 \alpha^{-4} K^{-1}\,.
\end{equation}

With the above choice of the function $f_x$ (depending on $\alpha$ and  the parameter $K$) we denote $U$  by $U_{K,\alpha}$ from now on.
With the aid of the previous estimates, we can now prove the following proposition. 
Its proof follows along similar lines as the corresponding argument for $\Omega=\R^3$ in \cite{griesemer}. 

\begin{proposition}\label{lem:GT}
For any $\eps>0$ there are $K>0$ and $C>0$ such that for all  $\alpha\gtrsim 1$ and any $\Psi \in L^2(\Omega)\otimes \F$ in the domain of $p^2 + \N$ 
\begin{equation}\label{eq:GT}
(1+ \eps) \| (p^2 + \N) \Psi\| + C \|\Psi\| \geq \| U_{K,\alpha} \Hh U_{K,\alpha}^\dagger \Psi\| \geq (1-\eps) \| (p^2 + \N) \Psi\| - C \|\Psi\|\,.
\end{equation}
\end{proposition}

We remark that due to the singular nature of $v_x$ in the interaction term, it is essential to apply the unitary transformation $U_{K,\alpha}$. In its absence, the bound \eqref{eq:GT} fails to hold. In other words, the domain of $\Hh$ does not coincide with the domain of $p^2 + \N$, but the one of $U_{K,\alpha} \Hh U_{K,\alpha}^\dagger$ does for a suitable choice of $K$. 

\begin{proof}
From \eqref{uhu} we see that the terms to estimate are the following:
\begin{align}\nonumber
\alpha^4 \| \left( a^\dagger( p f_x) + a(p f_x) \right)^2 \Psi\| & \lesssim \alpha^4 \sup_x \| p f_x \|_2^2 \| (\N + \alpha^{-2}) \Psi \| \\& \lesssim K^{-1}  \| (\N + \alpha^{-2}) \Psi \|
\end{align}
where we used \eqref{bound:pf}, 
\begin{align}\nonumber
\| \left( a ( g_x) + a^\dagger(g_x)  \right)  \Psi\| & \lesssim \sup_x  \| g_x \|_2 \| (\sqrt{\N + \alpha^{-2}} \Psi \| \\& \lesssim \delta  \| (\N + \alpha^{-2}) \Psi \| +  \delta^{-1} K \|\Psi\|
\end{align}
for any $\delta>0$, using \eqref{bound:g},  
\begin{align}\nonumber
\alpha^2 \| a^\dagger(p f_x) \cdot p \Psi\| &\leq \alpha^2 \sup_x \|p f_x\|_2 \| \sqrt{\N+\alpha^{-2}} \sqrt{p^2} \Psi\| \\ & \lesssim K^{-1/2} \| (p^2 + \N+\alpha^{-2})\Psi\|  \label{t3}
\end{align}
and finally, the term 
\begin{equation}\label{tt1}
\alpha^2  p\cdot a(p f_x)=  \alpha^2   a(p f_x)  \cdot p  + a(\alpha^2 p^2 f_x)\,.
\end{equation}
The first term on the right side of \eqref{tt1} can be estimated as in \eqref{t3} above. For the second term, we write 
\begin{equation}
\alpha^2 (p^2 f_x)(y)  = h^{(1)}_x(y) + h^{(2)}_x(y) 
\end{equation}
where 
\begin{equation}
h^{(1)}_x(y) = g_x(y) - f_x(y) + \left[  (-\Delta_\Omega)^{-1/2} - (K^2 -\Delta_\Omega)^{-1/2}\right](x,y)
\end{equation}
and
\begin{equation}
 h^{(2)}_x(y) = (K^2 - \Delta_\Omega)^{-1/2}(x,y)\,.
\end{equation}
The $L^2$-norms of $g_x$ and $f_x$ have already been bounded above, in \eqref{bound:g} and \eqref{bound:f}, respectively.  To bound  the third function in $h^{(1)}_x$, we use $0\leq t^{-1/2} - (K^2+t)^{-1/2} \leq K t^{-1/2} (K^2+t)^{-1/2}$, and find that the square of its  $L^2$-norm is bounded by 
\begin{equation}
 \frac 1{(2\pi)^3}  \int_{\R^3}   \frac {K^2} {k^2 \left( k^2 + K^2 \right) }  dk  = \frac 1{4\pi}  K\,.
\end{equation}
By using the Schwarz inequality we conclude that 
\begin{equation}
\left\| a(h_x^{(1)}) \Psi \right\| \lesssim \delta  \left\| \N \Psi \right\| +  \delta^{-1} K \left( 1+ (K\alpha)^{-4} \right) \| \Psi\|
\end{equation}
for any $\delta>0$. 

The last term to estimate is $a(h_x^{(2)})\Psi$. Since $| h_x^{(2)}(y) | \leq (K^2 - \Delta_{\R^3})^{-1/2}(x,y)$,  Lemma~\ref{lem:FS} implies that 
\begin{equation}
 \|  a(h_x^{(2)}) \Psi\| \leq (2\pi)^{-3/2}  \left(  \int_{\R^3} (K^2 + q^2)^{-1} |q|^{-2}  \, dq \right)^{1/2} \| \sqrt{\N}\sqrt{p^2} \Psi \|\,.
\end{equation}
The prefactor on the right side is equal to a constant times $K^{-1/2}$. Moreover, we can bound $\| \sqrt{\N}\sqrt{p^2} \Psi \| \leq \frac 12 \|(p^2+\N) \Psi \|$. In combination with \eqref{bound:f} and \eqref{bound:vf}, we hence  arrive at the desired result, with $K\sim \epsilon^{-2}$ and $C\sim \epsilon^{-3}$. 
\end{proof}

From Proposition~\ref{lem:GT} we draw two important conclusions. First, the ground state energy of $\Hh$ is uniformly bounded in $\alpha$, for large $\alpha$. Second, in any state of bounded energy, in the sense that $\| \Hh \Psi\| \leq {\rm const.}$, both $\|U_{K,\alpha}^\dagger p^2 U_{K,\alpha} \Psi\|$ and $\|U_{K,\alpha}^\dagger \N U_{K,\alpha} \Psi\|$ are uniformly bounded (for suitable $K$ independent of $\alpha$). In particular, we conclude that in order to compute the ground state energy, 
it  suffices to consider wave functions $\Psi$ having this property.

We have, by a similar computation as in \eqref{eq:unabla},
\begin{equation}\label{def:A}
U_{K,\alpha}^\dagger p^2 U_{K,\alpha} = (p - A_{K,\alpha})^2 \quad \text{with} \quad A_{K,\alpha} = \alpha^2 \left( a^\dagger( p f_x) + a(p f_x) \right)
\end{equation}
and 
\begin{equation}
U_{K,\alpha}^\dagger \N U_{K,\alpha} = \N - a(f_x) - a^\dagger(f_x) + \|f_x\|_2^2\,.
\end{equation}
Since $\|f_x\|_2$ is uniformly bounded, as shown in \eqref{bound:f} above, it easily follows that uniform boundedness of $\|U_{K,\alpha}^\dagger \N U_{K,\alpha} \Psi\|$ is equivalent to the one of $\| \N \Psi\|$.


\section{Proof of Theorem~\ref{main:thm}: Lower Bound}\label{sec:lower}

\subsection{Ultraviolet Cutoff}
The first step in the lower bound is to  introduce an ultraviolet cutoff in the interaction. Corollary~\ref{cor:LY} together with 
 Proposition~\ref{lem:GT} will allow us to quantify its effect on the ground state energy. 
 
 \begin{proposition}\label{prop:uvc}
For $\Lambda>0$,  let 
\begin{equation}\label{def:hamL}
\Hh^\Lambda = -\Delta_\Omega - a(v^\Lambda_x) - a^\dagger(v^\Lambda_x) + \N
\end{equation}
where
\begin{equation}\label{def:vL}
v_x^{\Lambda}(y) =  \frac{ \theta(\Lambda^2 + \Delta_\Omega)}{ (-\Delta_\Omega)^{1/2}}(x,y)\,.
\end{equation}
Then 
\begin{align}\nonumber
& \infspec \Hh - \infspec \Hh^\Lambda \\ & \quad  \gtrsim -   \Lambda^{-5/2} (\ln \Lambda)^{5/4} - \alpha^{-1} \Lambda^{-3/2} (\ln \Lambda)^{3/4} - \alpha^{-2} \Lambda^{-1} (\ln \Lambda)^{1/2}   \label{lb:co}
\end{align}
for $\alpha\gtrsim 1$ and $\Lambda\gtrsim 1$. 
 \end{proposition}

In order for the error introduced in \eqref{lb:co} to be negligible compared to $\alpha^{-2}$, it is sufficient to choose  $\Lambda \sim  \alpha^{\kappa}$ with $\kappa > 4/5$.

\begin{proof}
{\em Step 1.} Recall that  $v_x(y) = (-\Delta_\Omega)^{-1/2}(x,y)$. We pick some $0 < \Lambda' < \Lambda$ and  decompose $v_x$ as $v_x(y) = u_x^{\Lambda'}(y) + w_x(y)$ where $w_x$ is  defined as in \eqref{def:w} above, but with $\Lambda$ replaced by $\Lambda'$. I.e., 
\begin{equation}\label{def:w'}
w_x(y) = z(-\Delta_\Omega)(x,y) \quad \text{for} \quad z(t) = t^{-1/2} \left( 1 - e^{-t/\Lambda'^2} \right)^2\,.
\end{equation}
Corollary~\ref{cor:LY} states that 
\begin{equation}
a(w_x) + a^\dagger(w_x) \lesssim    \left( p^2 + \N + 1  \right)^2 \left( \Lambda'^{-5/2}   +  \alpha^{-1}  \Lambda'^{-3/2}  \right) 
\end{equation}
for $\alpha\gtrsim 1$. 
We now apply the unitary Gross transformation \eqref{def:U}, with $f_x$ given in \eqref{def:f}, and $K$ chosen such that Proposition~\ref{lem:GT} holds for some fixed $0<\epsilon<1$, say $\epsilon=1/2$. We have  
\begin{equation}
U^\dagger_{K,\alpha} a(w_x) U_{K,\alpha} = a(w_x) + \langle w_x | f_x \rangle
\end{equation}
and
\begin{equation}
\sup_{x\in\Omega} | \langle w_x|f_x \rangle | \lesssim \alpha^{-2} \Lambda'^{-1}
\end{equation}
which can easily be seen by noting that  $\langle w_x|f_x \rangle = (z \eta)(-\Delta_\Omega)(x,x)$ (with $z$ and $\eta$ defined in \eqref{def:w'} and \eqref{def:f}, respectively) and using that $|z(t) \eta(t)| \lesssim \alpha^{-2} (t+ \Lambda^2)^{-2}$, proceeding as in \eqref{eq:b2} to bound the expression in terms of the one for $\Omega=\R^3$. 
Proposition~\ref{lem:GT} thus implies that 
\begin{equation}
a(w_x) + a^\dagger(w_x) \lesssim  \left( \Hh + C  \right)^2 \left( \Lambda'^{-5/2}   + \alpha^{-1} \Lambda'^{-3/2} +  \alpha^{-2}  \Lambda'^{-1}  \right) 
\end{equation}
for a suitable constant $C>0$ (independent of $\alpha$ for $\alpha\gtrsim 1$). 

For computing the ground state energy, it is clearly sufficient to consider wave functions in the spectral subspace of $\Hh$ corresponding to $|\Hh| \leq C$ for a suitable constant $C$. We thus conclude that   
\begin{equation}\label{lb:s1}
\infspec \Hh \geq \infspec \tilde \Hh^{\Lambda'} - \const \left( \Lambda'^{-5/2}   + \alpha^{-1} \Lambda'^{-3/2} +  \alpha^{-2}  \Lambda'^{-1}  \right) 
\end{equation}
where $\tilde\Hh^{\Lambda'}$ is obtained from $\Hh$ by replacing $v_x$ with $u_x^{\Lambda'} = v_x - w_x$, i.e., 
\begin{equation}
u_x^{\Lambda'}(y) =  (-\Delta_\Omega)^{-1/2} \left( 1- \left ( 1 - e^{\Delta_\Omega/\Lambda'^2} \right)^2   \right)(x,y) \,.
\end{equation}

{\em Step 2.} 
We shall now  further truncate $u_x^{\Lambda'}$, and replace it by 
\begin{equation}
\tilde v_x^{\Lambda}(y) =  \frac{ \theta(\Lambda^2 + \Delta_\Omega)}{ (-\Delta_\Omega)^{1/2}}
 \left( 1- \left ( 1 - e^{\Delta_\Omega/\Lambda'^2} \right)^2   \right)(x,y) \,.
\end{equation} 
With the aid of \eqref{eq:b5}, one checks that 
\begin{equation}
\sup_{x\in\Omega} \| u_x^{\Lambda'} - \tilde v_x^\Lambda\|_2^2 \lesssim  \Lambda\, e^{-(\Lambda/\Lambda')^2}
\end{equation}
and hence, using the fact that $\sqrt{\N}$ is uniformly bounded for states with bounded energy, the  error for introducing this additional cutoff is at most of order $ \Lambda^{1/2} \, e^{-(\Lambda/\Lambda')^2/2}$.

{\em Step 3.} 
Finally, we want to further simplify $\tilde v_x^\Lambda$ and replace it by $v_x^\Lambda$ in \eqref{def:vL}. 
We claim that the ground state energy can only decrease under this replacement. This is the content of the following lemma.

\begin{lemma}
Let $\{\varphi_j\}_{j=1}^N$ be a set of orthonormal functions in $L^2(\Omega)$, and let 
\begin{equation}
u_x(y)  = \sum_{j=1}^N \lambda_j \overline{\varphi_j(x)} {\varphi_j(y)}  \quad \text {for $\lambda_j\geq 0$, $1\leq j\leq N$.} 
\end{equation}
Then
\begin{equation}
e(\lambda_1,\dots,\lambda_N) = \infspec \left[ -\Delta_\Omega - a(u_x) - a^\dagger(u_x) + \N\right]
\end{equation}
is decreasing in each $\lambda_j$.
\end{lemma}

\begin{proof} We shall use a Perron--Frobenius type argument. Let $\Psi\in L^2(\Omega)\otimes {\mathcal F}$ be given by $\{ \psi_0(x),\psi_1(x,y_1), \psi_2(x,y_1,y_2) ,\dots \}$. We extend $\{\varphi_j\}_{j=1}^N$ to an orthonormal basis $\{\varphi_j\}_{j\in \N}$ of $L^2(\Omega)$, and define $a^n_{i_1,\dots,i_n}(x)$ by the expansion $\psi_n(x,y_1,\dots,y_n) = \sum_{i_1,\dots,i_n} a^n_{i_1,\dots,i_n}(x) \varphi_{i_1}(y_1) \cdots \varphi_{i_n}(y_n)$. Then 
\begin{equation}
\langle \Psi | -\Delta_\Omega + \N | \Psi \rangle = \sum_{n\geq 0} \sum_{i_1,\dots,i_n} \left( \int_\Omega | \nabla_x a^n_{i_1,\dots,i_n}(x)|^2 dx + n  \int_\Omega |  a^n_{i_1,\dots,i_n}(x)|^2 dx \right)
\end{equation}
and 
\begin{equation}
\langle \Psi |  a(u_x) + a^\dagger(u_x)  | \Psi \rangle = 2  \sum_{j=1}^N  \lambda_j \sum_{n\geq 0} \sqrt{n+1}\sum_{i_1,\dots,i_n} \Re \int_\Omega \overline{a^n_{i_1,\dots,i_n}(x)}a^{n+1}_{i_1,\dots,i_n,j}(x) \varphi_j(x) \, dx \,.
\end{equation}
By multiplying the functions $a^n_{i_1,\dots,i_n}$ with an appropriate phase factor, we can make sure that 
\begin{equation}
\int_\Omega \overline{a^n_{i_1,\dots,i_n}(x)}a^{n+1}_{i_1,\dots,i_n,j}(x) \varphi_j(x) \, dx \geq 0
\end{equation}
for all $n\geq 0$, $1\leq j\leq N$ and all $i_1,\dots,i_n$, 
and this can  clearly only decrease the energy. When computing the ground state energy, it suffices to consider $\Psi$s with such property, in which case the energy is clearly monotone decreasing in all the $\lambda_j$.
\end{proof}

As a consequence, the ground state energy with interaction $\tilde v_x^{\Lambda}$ is bounded below by the one with interaction $v_x^{\Lambda}$. In particular, we have thus shown that
\begin{equation}\label{lb:s2o}
\infspec \Hh \geq \infspec \Hh^\Lambda - \const \left( \Lambda'^{-5/2} + \alpha^{-1} \Lambda'^{-3/2} +  \alpha^{-2} \Lambda'^{-1}+  \Lambda^{1/2} e^{-(\Lambda/\Lambda')^2 /2 }  \right)
\end{equation}
and this holds for all  $\alpha\gtrsim 1$ and $\Lambda'\gtrsim 1$. 
The choice $\Lambda' = \Lambda (6 \ln \Lambda)^{-1/2}$ yields \eqref{lb:co}. 
\end{proof}

\subsection{Final Lower Bound} 
The starting point of the proof of the lower bound is Proposition~\ref{prop:uvc}, which quantifies the error in replacing $\Hh$ by $\Hh^\Lambda$ in \eqref{def:hamL} for computing the ground state energy. We are thus left with giving a lower bound on $\infspec \Hh^\Lambda$. 

We choose, for simplicity, $\Lambda$ in such a way that $\Lambda^2$ is not an eigenvalue of $-\Delta_\Omega$. Let $\Pi$ denote the projection 
\begin{equation}
\Pi = \theta(\Lambda^2+\Delta_\Omega)
\qquad\text{and}\qquad
N = \dim\ran\Pi\,.
\end{equation}
For later purposes we note that one has the Weyl asymptotics
\begin{equation}\label{eq:weyl}
N \sim (2\pi)^{-3} |\Omega| \Lambda^3
\qquad\text{as}\ \Lambda\to\infty
\end{equation}
(see, e.g., \cite[Sec.~XIII.15]{RS4}). If $e_n$ and $\varphi_n$, respectively, denote the eigenvalues and (real-valued) eigenfunctions of $-\Delta_\Omega$, then
\begin{equation}
v_x^\Lambda(y) =  \sum_{n=1}^N \frac 1{\sqrt{e_n}} \varphi_n(x) \varphi_n(y) 
\end{equation}
has finite rank. The Fock space $\F(L^2(\Omega))$ naturally factors into a tensor product $\F(\Pi L^2(\Omega)) \otimes \F((\id-\Pi)L^2(\Omega))$, and $\Hh^\Lambda$ is of the form $\Aa\otimes \id + \id \otimes \N^>$, where $\Aa$ acts on $L^2(\Omega)\otimes \F(\Pi L^2(\Omega))$ and $\N^> = \sum_{n>N} a^\dagger(\varphi_n)a(\varphi_n)$ is the number operator on $\F((\id-\Pi)L^2(\Omega))$. In particular, $\infspec \Hh^\Lambda = \infspec \Aa$. 

As in Section~\ref{sec:upper} (where a different basis was used, however), we identify $\F(\Pi L^2(\Omega))$ with $L^2(\R^N)$ via the representation 
\begin{equation}\label{eq:ident}
 \varphi = \Pi \varphi = \sum_{n=1}^N \lambda_n \varphi_n \,,
\end{equation}
thus identifying a function $\varphi \in \ran \Pi$ with a point $(\lambda_1,\dots,\lambda_N) \in \R^N$. In this representation, we have
\begin{equation}
\Aa = -\Delta_\Omega + V_\varphi(x)  + \sum_{n=1}^N \left(  - \frac 1{4 \alpha^{4}} \partial_{\lambda_n}^2 +  \lambda_n^2  - \frac{1}{2 \alpha^{2}} \right)  
\end{equation}
on $L^2(\Omega)\otimes L^2(\R^N)$. For a lower bound, we can replace $-\Delta_\Omega + V_\varphi(x)$ by the infimum of its spectrum, for any fixed $\varphi\in\ran\Pi$. In particular,  we have
\begin{equation}\label{hek}
\infspec \Hh^\Lambda \geq \infspec \mathbb K
\end{equation}
where $\mathbb K$ is the operator on $L^2(\R^N)$ 
\begin{equation}\label{fmm}
\mathbb K = - \frac 1{4 \alpha^{4}} \sum_{n=1}^N \partial_{\lambda_n}^2  - \frac{N}{2 \alpha^{2}}  + \F^{\rm P} (\varphi)
\end{equation}
with $ \F^{\rm P}$ defined in \eqref{def:pekarF}. Here $ \F^{\rm P}(\varphi)$ is a function of $(\lambda_1,\dots,\lambda_N)$ via the identification \eqref{eq:ident}.

We now introduce an IMS type localization. Let $\chi:\R_+\to [0,1]$ be a smooth function with $\chi(t) =1$ for $t\leq 1/2$, $\chi(t)=0$ for $t\geq 1$. Let $\eps>0$, and let $j_1$ and $j_2$ denote the multiplication operators in $L^2(\R^N)$
\begin{equation}
j_1 = \chi( \epsilon^{-1} \| (-\Delta_\Omega)^{-1/2}  (\varphi -\pp)\|_2) \,, \  j_2 = \sqrt{ 1 - \chi( \epsilon^{-1} \| (-\Delta_\Omega)^{-1/2}  (\varphi -\pp)\|_2)^2}\,.
\end{equation}
Then clearly $j_1^2+j_2^2=1$ and
\begin{equation}
\mathbb{K} = j_1  \mathbb{K}  j_1 +  j_2 \mathbb{K}  j_2 - \mathbb{E} 
\end{equation}
where $\mathbb{E}$ is the IMS localization error
\begin{equation}
\mathbb{E} = \frac 1{4\alpha^{4}} \sum_{n=1}^N \left( |\partial_{\lambda_n} j_1|^2 + |\partial_{\lambda_n} j_2|^2 \right)\,.
\end{equation}
It is easy to see that  that $\mathbb{E} \lesssim \alpha^{-4} \epsilon^{-2}$, independently of $N$. In particular, the localization error is negligible if $ \eps \gg \alpha^{-1}$. 

On the support of $j_1$, we can use the  bound \eqref{lbh} on $\F^{\rm P}$. This gives 
\begin{equation}
j_1 \mathbb{K} j_1 \geq j_1^2 \infspec \left(  e^{\rm P} - \frac 1{4\alpha^{4}} \sum_{n=1}^N \partial_{\lambda_n}^2  - \frac N {2\alpha^{2}}  +  \langle \varphi- \pp |  \id - K - \epsilon C L | \varphi - \pp \rangle
 \right) 
\end{equation}
for $C$ a positive constant. 
Now $\pp$ will not necessarily be in the range of $\Pi$. However, since $\id - K -  \epsilon C L$ is positive for $\eps$ small enough, we can replace $\pp$ by its closest point (in the norm defined via $\id-K-\eps C L$) in the range of $\Pi$ for a lower bound. That is, 
\begin{equation}
\langle \varphi- \pp |  \id - K - \epsilon C L | \varphi - \pp \rangle \geq \langle \varphi-  y | \Pi( \id - K - \epsilon C L)\Pi | \varphi  -y\rangle
\end{equation}
where $y=  (\Pi (\id-K-\epsilon C L) \Pi)^{-1} \Pi (\id-K-\epsilon C L) \pp$. The shift by $y$ can be removed by a unitary transformation, without affecting the ground state energy. Hence 
\begin{align}\nonumber
j_1 \mathbb{K} j_1 & \geq j_1^2 \infspec \left(  e^{\rm P} - \frac 1{4\alpha^{4}} \sum_{n=1}^N \partial_{\lambda_n}^2  - \frac N {2\alpha^{2}}  +  \langle \varphi|  \Pi(\id - K - \epsilon C L )\Pi | \varphi \rangle
 \right) \\ & = j_1^2  \left[ e^{\rm P} - \frac 1{2 \alpha^{2}} \Tr  \left( \id- \sqrt {\id-\Pi(K+\eps C L)\Pi} \right) \right]\,.
\end{align}
This is of the correct form if $N\to \infty$ and $\eps\to 0$ as $\alpha\to \infty$. 

On the support of $j_2$, we use the bound \eqref{lb2} instead. We have, for any $\eta\geq 0$, 
\begin{align}\nonumber
j_2 \mathbb{K} j_2 & \geq j_2^2 \infspec \left(  e^{\rm P} - \frac 1{4\alpha^{4}} \sum\nolimits_{n=1}^N \partial_{\lambda_n}^2  - \frac N {2\alpha^{2}} +   \frac \eta 4   \eps^2 \right. \\ & \left.  \qquad\qquad\quad \phantom{\frac N2}+ \left\langle \varphi-\pp  \left| \id - \left(\id+\kappa' (-\Delta_\Omega)^{1/2} \right)^{-1} - \eta (-\Delta_\Omega)^{-1} \right|  \varphi - \pp \right\rangle  \right) 
\end{align}
where we have used the fact that $\| (-\Delta_\Omega)^{-1/2}  (\varphi -\pp)\|_2\geq \eps/2$ on the support of $j_2$. We choose $\eta$ independent of $\alpha$ (and hence also independent of $\Lambda$ and $\epsilon$) and small enough such that the operator in the last line is positive. Proceeding as in the case of $j_1$ above, we obtain
\begin{align}
j_2 \mathbb{K} j_2 & \geq j_2^2 \left( e^{\rm P} +   \frac \eta 4 \eps^2 -  \frac 1{2\alpha^{2}} \Tr \Pi  \left[ \id - \sqrt{ \id -  \left(1+\kappa' (-\Delta_\Omega)^{1/2} \right)^{-1} - \eta (-\Delta_\Omega)^{-1} } \right] \right) \,. \label{731}
\end{align}
From the Weyl asymptotics \eqref{eq:weyl} one checks that the trace diverges like $N^{2/3} \sim  \Lambda^2$ for large $\Lambda$. 
Hence if we choose $ \Lambda\alpha^{-1}\leq\const \epsilon$ with a sufficiently small constant, the term in parenthesis in \eqref{731} is actually larger than $e^{\rm P}$. Since we will choose  $\Lambda\sim \alpha^\kappa$ with $\kappa > 4/5$, this is compatible with the condition $\eps \ll 1$ as long as $\kappa <1$.

We thus conclude that if $\Lambda\alpha^{-1}\leq\const \epsilon$ and $\epsilon$ is small enough, we have the  bound 
\begin{equation}
\infspec \mathbb K \geq  e^{\rm P} - \frac 1{ 2\alpha^{2}} \Tr  \left( \id- \sqrt {\id-\Pi(K+\eps C L)\Pi} \right) - \const \alpha^{-4} \epsilon^{-2} \,.
\end{equation}
For a lower bound, we can further drop the $\Pi$'s in the second term on the right side, and replace them by $\id$. Note that $\|K + \epsilon C L \| \leq \nu < 1$ for small enough $\epsilon$, and the function $f(t) = 1 - \sqrt{1-t}$ is Lipschitz continuous and convex on $[0,\nu]$. We utilize the following simple lemma.

\begin{lemma}\label{lem:Lip}
For $\nu>0$, let $f:[0,\nu]\to \R$ be a Lipschitz-continuous and convex function with $f(0)=0$, and let $A,B$ be non-negative trace-class operators with $A+ B \leq \nu$. Then
\begin{equation}
\Tr f(A+B) \leq \Tr f(A) + C_f \Tr B
\end{equation}
where $C_f$ denotes the Lipschitz constant of $f$.
\end{lemma}

\begin{proof}
With $\{g_j\}$ a basis of eigenvectors of $A+B$, we have
\begin{equation}
\Tr f(A+B) = \sum_j f ( \langle g_j | A+B| g_j\rangle ) \leq \sum_j f ( \langle g_j | A | g_j\rangle )  + C_f \sum_j \langle g_j | B | g_j\rangle \,.
\end{equation}
The convexity of $f$ implies that $f ( \langle g_j | A | g_j\rangle ) \leq  \langle g_j | f(A) | g_j\rangle$, which yields the desired result.
\end{proof}

Lemma~\ref{lem:Lip} readily implies  that 
\begin{equation}
\Tr  \left( \id- \sqrt {\id-K-\eps C L} \right) \leq \Tr  \left( \id- \sqrt {\id-K} \right) + \const \epsilon\, \Tr L\,.
\end{equation}
We thus have
\begin{equation}\label{lb:e2}
\infspec \mathbb K \geq  e^{\rm P} - \frac 1{2\alpha^{2}} \Tr  \left( \id- \sqrt {\id-K} \right) - \const \left( \alpha^{-4} \epsilon^{-2} + \alpha^{-2} \epsilon\right)\,.
\end{equation}
In combination with \eqref{lb:co} and \eqref{hek}, this is our final lower bound. 

To minimize the error terms in \eqref{lb:co} and \eqref{lb:e2}, we shall choose $\epsilon\sim \alpha^{-1/7}(\ln \alpha)^{5/14}$ and $\Lambda \sim \alpha^{6/7} (\ln \alpha)^{5/14}$. 
This yields
\begin{equation}
\infspec \Hh \geq e^{\rm P} - \frac 1{2\alpha^{2}} \Tr  \left( \id- \sqrt {\id-K} \right) - \const \alpha^{-15/7} (\ln \alpha)^{5/14}
\end{equation}
for  $\alpha\gtrsim 1$, and thus completes the proof of the lower bound in Theorem~\ref{main:thm}.
\hfill\qed


\appendix

\section{Equivalent Formulation of Assumption~\ref{ass:hessian}}\label{sec:compact}

In this appendix we shall explain how Assumption~\ref{ass:hessian} can be verified  via a spectral analysis of the Hessian of $\E^{\rm P}$ at its minimizer $\psi^{\rm P}\geq 0$, which is assumed to be unique. We partly follow ideas in \cite[Sec.~2]{FLS}. 

The Euler--Lagrange equation for the minimizer is
\begin{equation}\label{ELeq}
-\Delta_\Omega \psi^{\rm P} - 2 \left( (-\Delta_\Omega)^{-1}|\psi^{\rm P}|^2\right)\psi^{\rm P} = \mu \psi^{\rm P} \,.
\end{equation}
The relevant Hessian $Z^{\rm P}$ is defined via
\begin{equation}
\langle \psi | Z^{\rm P} | \psi \rangle = \lim_{\epsilon\to 0} \frac 1 {\epsilon^{2}} \left(  \E^{\rm P} \left (  \frac { \psi^{\rm P} + \epsilon \psi }{ \| \psi^{\rm P}  + \epsilon \psi\|_2} \right) - e^{\rm P} \right) 
\end{equation}
for real-valued $\psi \in H_0^1(\Omega)$, and 
 equals 
\begin{align}\nonumber
Z^{\rm P} & = -\Delta_\Omega - 2 (-\Delta_\Omega)^{-1}|\psi^{\rm P}|^2 - 4 X - \mu \\ \nonumber
& \quad - 4 \iint_{\Omega\times\Omega} |\psi^{\rm P}(x)|^2 (-\Delta_\Omega)^{-1}(x,y)|\psi^{\rm P}(y)|^2\,dx\,dy \  |\psi^{\rm P}\rangle\langle\psi^{\rm P}| \\
& \quad + 4 \left( |\psi^{\rm P}\rangle\langle \left( (-\Delta_\Omega)^{-1}|\psi^{\rm P}|^2\right)\psi^{\rm P} | + \text{h.c.} \right)
\end{align}
where $X$ is the operator with integral kernel
$$
X(x,y) = \psi^{\rm P}(x)(-\Delta_\Omega)^{-1}(x,y)\psi^{\rm P}(y)\,.
$$
There is also another Hessian defined for purely imaginary perturbations of $\psi^{\rm P}$, but it is trivially given by the linear operator defined by the equation~\eqref{ELeq} and plays no role here.

Note that $Z^{\rm P} \psi^{\rm P}=0$. We now show that if $\psi^{\rm P}$ spans the kernel of $Z^{\rm P}$, then Assumption~\ref{ass:hessian} holds.

\begin{lemma}
If $\ker Z^{\rm P} = \spa \{\psi^{\rm P}\}$ then 
there exists a $\kappa>0$ such that for all $0\leq\psi\in H^1_0(\Omega)$ with $\|\psi\|_2=1$ we have
\begin{equation}\label{A11}
\mathcal E^{\rm P}(\psi) \geq \mathcal E^{\rm P}(\psi^{\rm P}) + \kappa \|\psi-\psi^{\rm P}\|_{H^1(\Omega)}^2 \,.
\end{equation}
\end{lemma}

\begin{proof}
\emph{Step 1.} We first show that there are $c>0$ and $\kappa>0$ such that \eqref{A11} holds  for all $0\leq\psi\in H^1_0(\Omega)$ with $\|\psi\|_2=1$ and $\|\psi-\psi^{\rm P}\|_{H^1(\Omega)}\leq c$. 
We set $\delta = \psi-\psi^{\rm P}$ and expand
\begin{align}\nonumber
\mathcal E^{\rm P}(\psi^{\rm P}+\delta) & = \mathcal E^{\rm P}(\psi^{\rm P}) + 2\mu \int_\Omega \psi^{\rm P}(x)\delta(x)\,dx \\ \nonumber
& \quad + \int_\Omega |\nabla\delta(x)|^2\,dx - 2 \iint_{\Omega\times\Omega} \psi^{\rm P}(x)^2(-\Delta_\Omega)^{-1}(x,y)\delta(y)^2\,dx\,dy
- 4\langle\delta|X|\delta\rangle \\
& \quad + O(\|\delta\|_{H^1}^3) \,.
\end{align}
The assumption $\|\psi\|_2=1$ implies that
\begin{equation}
2 \int_\Omega \psi^{\rm P}(x)\delta(x)\,dx = - \|\delta\|^2_2 \,,
\end{equation}
and therefore, using this identity multiple times,
\begin{equation}
\mathcal E^{\rm P}(\psi^{\rm P}+\delta)   = \mathcal E^{\rm P}(\psi^{\rm P}) + \langle \delta|Z^{\rm P}|\delta\rangle+ O(\|\delta\|_{H^1}^3) \,.
\end{equation}
The operator $Z^{\rm P}$ has discrete spectrum, and hence our assumption on the simplicity of the kernel  implies that for some $\kappa>0$
\begin{equation}
\langle \delta|Z^{\rm P}|\delta\rangle \geq \kappa \| \delta - \langle\psi^{\rm P}|\delta\rangle\psi^{\rm P}\|_2^2 = \kappa \left( \|\delta\|_2^2 - \left( \int_\Omega \psi^{\rm P}\delta \right)^2 \right) = \kappa \|\delta\|_2^2 \left(1- 4^{-1} \|\delta\|_2^2 \right)\,.
\end{equation}
On the other hand, it is easy to see that for some $C>0$
\begin{equation}
Z^{\rm P} \geq -(1/2)\Delta_\Omega - C \,.
\end{equation}
Taking a mean of the previous two inequalities we obtain for any $0\leq\theta\leq 1$,
\begin{equation}
\langle \delta|Z^{\rm P}|\delta\rangle \geq (\theta/2)\|\nabla\delta\|_2^2 + ((1-\theta)\kappa -C\theta)\|\delta\|_2^2 - 4^{-1} \kappa (1-\theta) \|\delta\|_2^4 \,.
\end{equation}
In particular, for $\theta = \kappa/(C+\kappa+1/2)$ we have
\begin{equation}
\langle \delta|Z^{\rm P}|\delta\rangle \geq \frac{\kappa}{2C+2\kappa+1} \|\delta\|_{H^1(\Omega)}^2 - 4^{-1} \kappa \frac{2C+1}{2C+2\kappa+1} \|\delta\|_2^4 \,.
\end{equation}
Inserting this into the above inequality, we obtain
\begin{equation}
\mathcal E^{\rm P}(\psi^{\rm P}+\delta) \geq \mathcal E^{\rm P}(\psi^{\rm P}) +  \frac{\kappa}{2C+2\kappa+1} \|\delta\|_{H^1(\Omega)}^2 + O(\|\delta\|_{H^1(\Omega)}^3) \,,
\end{equation}
which clearly implies the assertion in Step 1.

\emph{Step 2.} We now prove the full statement of the lemma. We argue by contradiction. If there were no such $\kappa$, we could find a sequence $0\leq\psi_n\in H^1_0(\Omega)$ with $\|\psi_n\|_2=1$ such that
\begin{equation}\label{a13}
\mathcal E^{\rm P}(\psi_n) < \mathcal E^{\rm P}(\psi^{\rm P}) + n^{-1} \|\psi_n-\psi^{\rm P}\|_{H^1(\Omega)}^2 \,.
\end{equation}
Using \eqref{eq:b2}, Hardy--Littlewood--Sobolev, H\"older and Sobolev we bound
\begin{align}\nonumber
\iint_{\Omega\times\Omega} \psi(x)^2 (-\Delta_\Omega)^{-1}(x,y)\psi(y)^2\,dx\,dy & \leq \frac{1}{4\pi} \iint_{\Omega\times\Omega} \frac{\psi(x)^2\psi(y)^2}{|x-y|}\,dx\,dy \lesssim \|\psi^2\|_{6/5}^2\\
& \leq \|\psi\|_6 \|\psi\|_2^3 \lesssim \|\nabla\psi\|_2 \|\psi\|_2^3 \,.
\end{align}
This implies $\mathcal E^{\rm P}(\psi) \geq (1/2)\|\nabla\psi\|_2^2 - C\|\psi\|_2^6$ for all $\psi\in H^1_0(\Omega)$. Combining this inequality with the upper bound \eqref{a13} on $\mathcal E^{\rm P}(\psi_n)$ we easily infer that $(\psi_n)$ is bounded in $H^1_0(\Omega)$ and hence that $\|\psi_n-\psi^{\rm P}\|_{H^1(\Omega)}$ is bounded. Thus, \eqref{a13} implies that $(\psi_n)$ is a minimizing sequence for $\mathcal E^{\rm P}$. Therefore, by a simple compactness argument, after passing to a subsequence, $\psi_n$ converges in $H^1$ to a minimizer. Since $\psi_n\geq 0$, our assumed uniqueness of the minimizer implies that $\psi_n\to\psi^{\rm P}$. Thus, for all sufficiently large $n$, $\|\psi_n-\psi^{\rm P}\|_{H^1(\Omega)}\leq c$, where $c$ is the constant from Step 1. Therefore the inequality from Step 1 is applicable, but this bound contradicts \eqref{a13} for large $n$. This completes the proof.
\end{proof}

\section{Bounds on Solutions of Poisson's Equation}\label{sec:appa}

We consider solutions $u$ of the equation $-\Delta u=f$ in an open set $\Omega\subset\R^d$ with boundary conditions $u=0$ on $\partial\Omega$. We are interested in bounds on derivatives of $u$ in terms of derivatives of $f$, uniformly on small balls, possibly intersecting the boundary of $\Omega$. While we use these bounds only for $d=3$, it requires no extra effort to prove them in arbitrary dimension $d\geq 2$. 

\subsection{Statement of the Inequality}

Let $k\in\N$ and $\delta\in (0,1)$. We say that an open set $\Omega\subset\R^d$ is a $C^{k,\delta}$ set if there are constants $r_0>0$ and $M<\infty$ such that for any $x\in\partial\Omega$ there is a function $\Gamma:\{y'\in\R^{d-1}:\ |y'|<r_0\} \to\R$ satisfying $\Gamma(0)=0$, $\nabla\Gamma(0)=0$ and 
\begin{equation}\label{eq:a1}
\sum_{j=0}^k r_0^{j-1} \sup_{|y'|<r_0} |\partial^j\Gamma(y')| + r_0^{k-1+\delta} \sup_{|y'|,|z'|<r_0} \frac{|\partial^k\Gamma(y')-\partial^k\Gamma(z')|}{|y'-z'|^\delta} \leq M
\end{equation}
such that, after a translation and a rotation (which maps $x$ to $0$ and the exterior unit normal at $x$ to $(0,\ldots,0,-1)$, and is denoted by $\mathcal{T}_x$),
\begin{equation}
\mathcal{T}_x \left( \Omega\cap B_{r_0}(x) \right) = \{ (y',y_d)\in\R^{d-1}\times\R :\ |y'|<r_0 \,,\ y_d>\Gamma(y') \}\cap B_{r_0}(0) \,.
\end{equation}
Here and below we use the notation $|\partial^k f(x)|= (\sum_{|\beta|=k}|\partial^\beta f(x)|^2)^{1/2}$ and similarly $|\partial^{k}f(x)-\partial^{k}f(y)| = (\sum_{|\beta|=k}|\partial^\beta f(x) - \partial^\beta f(y)|^2)^{1/2} $, with $\partial^\beta = \partial_1^{\beta_1} \cdots \partial_d^{\beta_d}$ for $\beta\in \N_0^d$, and $|\beta|=\sum_{j=1}^d \beta_j$. The above definition of a $C^{k,\delta}$ set is standard (see, e.g., \cite[Sec.~6.2]{GiTr}), except possibly for the choice of the $r_0$ dependence in \eqref{eq:a1}. Our choice ensures scale invariance in the sense that if $\Omega$ is scaled by a factor $\lambda$, $r_0$ gets multiplied by $\lambda$ while $M$ stays the same.

\begin{theorem}\label{pdebound}
Let $k\in\N$, $0<\delta<1$, $R_0 > 0$ and $\Omega\subset\R^d$ be an open $C^{k,\delta}$ set. Then  we have, for all $a\in\Omega$ and all $R\leq R_0$, if $k=1$
\begin{equation}
\sum_{j=0}^{1} R^{j} \sup_{B_{R}(a)\cap\Omega}|\partial^j u| + R^{1+\delta} \sup_{x,y\in B_{R}(a)\cap\Omega} \frac{|\partial u(x)-\partial u(y)|}{|x-y|^\delta}
\lesssim \sup_{B_{2R}(a)\cap\Omega}|u| + R^2 \sup_{B_{2R}(a)\cap\Omega}|f|
\end{equation}
 and if $k\geq 2$ 
\begin{align}\nonumber
& \sum_{j=0}^{k} R^{j} \sup_{B_{R}(a)\cap\Omega}|\partial^j u| + R^{k+\delta} \sup_{x,y\in B_{R}(a)\cap\Omega} \frac{|\partial^{k}u(x)-\partial^{k}u(y)|}{|x-y|^\delta} \\
& \quad \lesssim \sup_{B_{2R}(a)\cap\Omega}|u| + \sum_{j=0}^{k-2} R^{j+2} \sup_{B_{2R}(a)\cap\Omega}|\partial^j f| + R^{k+\delta} \sup_{x,y\in B_{2R}(a)\cap\Omega} \frac{|\partial^{k-2}f(x)-\partial^{k-2}f(y)|}{|x-y|^\delta} \,.
\end{align}
The constants in these bounds depend only on $d$, $k$, $\delta$, $M$ and $R_0/r_0$.
\end{theorem}

Dropping the H\"older semi-norm on the left side and estimating it on the right side in terms of one higher derivative, we obtain

\begin{corollary}\label{pdeboundcor}
Let $k\in\N$, $0<\delta<1$, $R_0 >0$ and $\Omega\subset\R^d$ be an open $C^{k,\delta}$ set. Then we have for all $a\in\Omega$ and all $R\leq R_0$,
\begin{equation}\label{eq:cora}
\sum_{j=0}^{k} R^{j} \sup_{B_{R}(a)\cap\Omega}|\partial^j u|  \lesssim \sup_{B_{2R}(a)\cap\Omega}|u| + \sum_{j=0}^{k-1} R^{j+2} \sup_{B_{2R}(a)\cap\Omega}|\partial^j f| \,.
\end{equation}
The constants in these bounds depend only on $d$, $k$, $\delta$, $M$ and $R_0/r_0$.
\end{corollary}


\subsection{Local Estimates}

The more difficult assertion in Theorem \ref{pdebound} is for balls such that $B_{2R}(a)\cap\partial\Omega\neq\emptyset$. The strategy in this case will be to flatten the boundary, but this results in a second order elliptic equation with variable coefficients. In this subsection we state and prove bounds on solutions of such equations for domains with a flat boundary portion. 

Let $\Omega\subset\R^d_+ \coloneqq \R^{d-1} \times (0,\infty)$ be an open set with an open boundary portion $T$ on $\partial\R^d_+$. We emphasize explicitly that the case $T=\emptyset$ is allowed. For $x,y\in\Omega$ we write, following \cite[Sect.~4.4]{GiTr},
\begin{equation}
\overline d_x \coloneqq\dist(x,\partial\Omega\setminus T) \,,
\qquad
\overline d_{x,y} \coloneqq \min\{\overline d_x,\overline d_y\} \,,
\end{equation}
and introduce the norms
\begin{equation}
|u|_{k,\Omega\cup T}^{(\sigma)} \coloneqq \sum_{j=0}^k \sup_{x\in\Omega} \overline d_x^{j+\sigma} |\partial^j u(x)|
\end{equation}
and
\begin{equation}\label{eq:a8}
|u|_{k,\delta,\Omega\cup T}^{(\sigma)} \coloneqq \sum_{j=0}^k \sup_{x\in\Omega} \overline d_x^{j+\sigma} |\partial^j u(x)| + \sup_{x,y\in\Omega} \overline d_{x,y}^{k+\delta+\sigma} \frac{|\partial^k u(x)-\partial^k u(y)|}{|x-y|^\delta} \,.
\end{equation}
One readily checks that these norms satisfy  $|f g|_{k,\delta,\Omega\cup T}^{(\sigma_1+\sigma_2)} \lesssim |f |_{k,\delta,\Omega\cup T}^{(\sigma_1)}| g|_{k,\delta,\Omega\cup T}^{(\sigma_2)}$ as well as $| \partial f|_{k,\delta,\Omega\cup T}^{(\sigma)}\lesssim |  f|_{k+1,\delta,\Omega\cup T}^{(\sigma-1)}$ and $|  f|_{k,\delta,\Omega\cup T}^{(\sigma)} \lesssim |  f|_{k+1,\delta,\Omega\cup T}^{(\sigma)}$ with implicit constants depending only on $d$, $k$, $\delta$ and $\sigma$.

The following two lemmas are the  main technical ingredients in the proof of Theorem~\ref{pdebound}.

\begin{lemma}\label{c1alpha}
Let $0<\delta<1$ and $\Omega\subset\R^d_+$ be an open set with a boundary portion $T$ on $\partial\R^d_+$. Let
\begin{equation}
Lu=f +\nabla\cdot g \qquad\text{in}\ \Omega
\qquad\text{and}\qquad
u=0 \qquad\text{on}\ T \,,
\end{equation}
where
\begin{equation}
L = 
- \sum_{r,s=1}^d \partial_r a_{r,s}\partial_s \,.
\end{equation}
Then
\begin{equation}
|u|_{1,\delta,\Omega\cup T}^{(0)} \lesssim |u|_{0,\Omega\cup T}^{(0)} + |f|_{0,\Omega\cup T}^{(2)} + |g|_{0,\delta,\Omega\cup T}^{(1)} \,,
\end{equation}
with the implicit constant depending only on $d$, $\delta$, $\lambda$ and $\Lambda$, where
\begin{equation}
\sum_{r,s=1}^d |a_{r,s}|_{0,\delta,\Omega\cup T}^{(0)} \leq \Lambda 
\end{equation}
and $\lambda>0$ is a uniform lower bound on the lowest eigenvalue of the symmetric matrix defined by $a_{r,s}$.
\end{lemma}

For us the bound with $g=0$ suffices, but $g$ appears naturally in the proof.

\begin{proof}
A similar, but less precise bound appears in \cite[Corollary 8.36]{GiTr}. Since its proof is sketched only very briefly, we provide some more details. The starting point is \cite[(4.46)]{GiTr}, which proves the lemma in the case $L=-\Delta$ and $\Omega =B_R(x_0)\cap\R^d_+$ with $x_0\in\R^d_+$. By the same argument as in the proof of \cite[Theorem 4.12]{GiTr} (which is not given, but which is similar to the proof of \cite[Theorem 4.8]{GiTr}), this bound leads to Lemma \ref{c1alpha} for $L=-\Delta$, but for general $\Omega$. Using a simple change of variables as in the proof of \cite[Lemma 6.1]{GiTr} we obtain the lemma for $L=-\nabla\cdot A\nabla$ with a constant matrix $A$ again for a general $\Omega$. Finally, using the perturbation argument as in the proof of \cite[Lemma 6.4]{GiTr} (which again is not given, but which is similar to the proof of \cite[Theorem 6.2]{GiTr}) we obtain the lemma.
\end{proof}

\begin{lemma}\label{c2alpha}
Let $k\geq 2$, $0<\delta<1$ and $\Omega\subset\R^d_+$ be an open set with a boundary portion $T$ on $\partial\R^d_+$. Let
\begin{equation}
Lu=f \qquad\text{in}\ \Omega
\qquad\text{and}\qquad
u=0 \qquad\text{on}\ T \,,
\end{equation}
where
\begin{equation}
L = - \sum_{r,s=1}^d a_{r,s} \partial_r\partial_s + \sum_{r=1}^d b_r\partial_r\,.
\end{equation}
Then
\begin{equation}
|u|_{k,\delta,\Omega\cup T}^{(0)} \lesssim |u|_{0,\Omega\cup T}^{(0)} + |f|_{k-2,\delta,\Omega\cup T}^{(2)}
\end{equation}
with the implicit constant depending only on $d$, $k$, $\delta$, $\lambda$ and $\Lambda$, where
\begin{equation}
\sum_{r,s=1}^d |a_{r,s}|_{k-2,\delta,\Omega\cup T}^{(0)} + \sum_{r=1}^d |b_r|_{k-2,\delta,\Omega\cup T}^{(1)}\leq \Lambda 
\end{equation}
and $\lambda>0$ is a uniform lower bound on the lowest eigenvalue of the symmetric matrix defined by $a_{r,s}$.

\end{lemma}

\begin{proof}
Lemma \ref{c2alpha} with $k=2$ coincides with \cite[Lemma 6.4]{GiTr}. Estimates similar to, but less precise than our statement for $k\geq 3$ are stated as \cite[Problem 6.2]{GiTr}, but without any details.

We shall show that for any integer $k\geq 2$ and any  $\sigma\geq 0$,
\begin{equation}
\label{eq:c1alphagen}
|u|_{k,\delta,\Omega\cup T}^{(\sigma)} \lesssim |u|_{0,\Omega\cup T}^{(\sigma)} + |f|_{k-2,\delta,\Omega\cup T}^{(\sigma+2)}
\end{equation}
where the implicit constant depends only on $d$, $k$, $\delta$, $\sigma$, $\lambda$ and $\Lambda$. We will prove this by induction on $k$.

First, let $k=2$. For $\sigma=0$ the claimed inequality is \cite[Lemma 6.4]{GiTr} (whose proof is not given, but which is similar to the proof of \cite[Theorem 6.2]{GiTr}). The proof for $\sigma>0$ follows by the same argument.

Now let $k\geq 3$ and $\sigma\geq 0$. We assume the inequality has already been shown for all smaller values of $k$ and for all values of $\sigma$. For $1\leq j\leq d-1$ the function $v=\partial_j u$ satisfies
\begin{equation}
Lv=\tilde f \qquad\text{in}\ \Omega
\qquad\text{and}\qquad
v=0 \qquad\text{on}\ T \,,
\end{equation}
where
\begin{equation}
\tilde f = \partial_j f + \sum_{r,s=1}^d (\partial_j a_{r,s}) \partial_r\partial_s u -\sum_{r=1}^d (\partial_j b_r)\partial_r u \,.
\end{equation}
Therefore, by the induction assumption \eqref{eq:c1alphagen} with $\sigma$ replaced by $\sigma+1$,
\begin{align}\nonumber
& |v|_{k-1,\delta,\Omega\cup T}^{(\sigma+1)} \\ \nonumber & \lesssim |v|_{0,\Omega\cup T}^{(\sigma+1)} + |\tilde f|_{k-3,\delta,\Omega\cup T}^{(\sigma+3)} \\
\nonumber & \leq |\partial_j u|_{0,\Omega\cup T}^{(\sigma+1)} + |\partial_j f|_{k-3,\delta,\Omega\cup T}^{(\sigma+3)} + \sum_{r,s} |(\partial_j a_{r,s}) \partial_r\partial_s u|_{k-3,\delta,\Omega\cup T}^{(\sigma+3)} + \sum_r |(\partial_j b_r)\partial_r u|_{k-3,\delta,\Omega\cup T}^{(\sigma+3)} \\
\nonumber& \lesssim |\partial_j u|_{0,\Omega\cup T}^{(\sigma+1)} + |\partial_j f|_{k-3,\delta,\Omega\cup T}^{(\sigma+3)} + \sum_{r,s} |\partial_j a_{r,s}|_{k-3,\delta,\Omega\cup T}^{(1)} |\partial_r\partial_s u|_{k-3,\delta,\Omega\cup T}^{(\sigma+2)} \\
\nonumber& \qquad\qquad\qquad+ \sum_r |\partial_j b_r|_{k-3,\delta,\Omega\cup T}^{(2)} |\partial_r u|_{k-3,\delta,\Omega\cup T}^{(\sigma+1)} \\
\nonumber& \lesssim |u|_{1,\Omega\cup T}^{(\sigma)} + |f|_{k-2,\delta,\Omega\cup T}^{(\sigma+2)} + \sum_{r,s} |a_{r,s}|_{k-2,\delta,\Omega\cup T}^{(0)} |u|_{k-1,\delta,\Omega\cup T}^{(\sigma)} + \sum_r |b_r|_{k-2,\delta,\Omega\cup T}^{(1)} |u|_{k-2,\delta,\Omega\cup T}^{(\sigma)} \\
& \lesssim |f|_{k-2,\delta,\Omega\cup T}^{(\sigma+2)} + |u|_{k-1,\delta,\Omega\cup T}^{(\sigma)} \,,
\end{align}
where we have used the properties of the norms discussed after Eq.~\eqref{eq:a8}. 
Bounding the last term on the right side using the induction assumption with $\sigma$, we finally obtain
\begin{equation}
\label{eq:c3alpha1}
|\partial_j u|_{k-1,\delta,\Omega\cup T}^{(\sigma+1)} \lesssim |u|_{0,\Omega\cup T}^{(\sigma)} + |f|_{k-2,\delta,\Omega\cup T}^{(\sigma+2)}
\qquad\text{if}\ j=1,\ldots, d-1 \,.
\end{equation}

On the other hand, we have
\begin{equation}
\partial_d^2 u = \frac{1}{a_{dd}} \left( - \sum_{(r,s)\neq(d,d)} a_{rs}\partial_r\partial_s u + \sum_r b_r \partial_r u - f \right)
\end{equation}
and therefore
\begin{align}\nonumber
|\partial_d^2 u|_{k-2,\delta,\Omega\cup T}^{(\sigma+2)} & \leq | a_{dd}^{-1}|_{k-2,\delta,\Omega\cup T}^{(0)} \left( \sum_{(r,s)\neq(d,d)} |a_{rs}|_{k-2,\delta,\Omega\cup T}^{(0)} |\partial_r\partial_s u|_{k-2,\delta,\Omega\cup T}^{(\sigma+2)} \right. \\
 & \qquad\qquad\qquad \left.
+ \sum_r |b_r|_{k-2,\delta,\Omega\cup T}^{(1)} |\partial_r u|_{k-2,\delta,\Omega\cup T}^{(\sigma+1)} + |f|_{k-2,\delta,\Omega\cup T}^{(\sigma+2)} \right) \,.
\end{align}
Our assumptions imply that $| a_{dd}^{-1}|_{k-2,\delta,\Omega\cup T}^{(0)}$ is bounded in terms of $\Lambda$ and $\lambda$. Moreover, $|\partial_r u|_{k-2,\delta,\Omega\cup T}^{(\sigma+1)}$ is bounded above  for any $1\leq r\leq d$ by $|u|_{k-1,\delta,\Omega\cup T}^{(\sigma)}$, which by the induction hypothesis  \eqref{eq:c1alphagen} is bounded by $|u|_{0,\Omega\cup T}^{(\sigma)} + |f|_{k-3,\delta,\Omega\cup T}^{(\sigma+2)}$. We thus conclude that
\begin{align}\nonumber
 |\partial_d^2 u|_{k-2,\delta,\Omega\cup T}^{(\sigma+2)} &  \lesssim \sum_{(r,s)\neq(d,d)} |\partial_r\partial_s u|_{k-2,\delta,\Omega\cup T}^{(\sigma+2)} + \sum_r |\partial_r u|_{k-2,\delta,\Omega\cup T}^{(\sigma+1)} + |f|_{k-2,\delta,\Omega\cup T}^{(\sigma+2)} \\
& \lesssim \sum_{j=1}^{d-1} |\partial_j u|_{k-1,\delta,\Omega\cup T}^{(\sigma+1)} + |f|_{k-2,\delta,\Omega\cup T}^{(\sigma+2)} \,.
\end{align}
Combining this with \eqref{eq:c3alpha1} we obtain the claimed estimate on $|u|_{k,\delta,\Omega\cup T}^{(\sigma)}$. This completes the proof  of Lemma \ref{c2alpha}. 
\end{proof}


\subsection{Proof of Theorem \ref{pdebound}}

We first assume that $\dist(a,\partial\Omega)\geq 2R$. In this case $B_{2R}(a)\subset\Omega$ and we can apply Lemmas \ref{c1alpha} and \ref{c2alpha} with $L=-\Delta$, $T=\emptyset$ and $B_{2R}(a)$ playing the role of $\Omega$. Since
\begin{equation}
|u|_{k,\delta,\Omega}^{(\sigma)} \leq \sum_{j=0}^k \sup_{x\in B_{2R}(a)} (2R)^{j+\sigma} |\partial^j u(x)| + \sup_{x,y\in B_{2R}(a)} (2R)^{k+\delta+\sigma} \frac{|\partial^k u(x)-\partial^k u(y)|}{|x-y|^\delta}
\end{equation}
(and similarly with $u$ replaced by $f$) 
and
\begin{equation}
|u|_{k,\delta,\Omega}^{(\sigma)} \geq \sum_{j=0}^k \sup_{x\in B_R(a)} R^{j+\sigma} |\partial^j u(x)| + \sup_{x,y\in B_R(a)} R^{k+\delta+\sigma} \frac{|\partial^k u(x)-\partial^k u(y)|}{|x-y|^\delta} \,,
\end{equation}
we immediately obtain the bound in this case. (Of course, in order to prove the bounds much simpler versions of Lemmas \ref{c1alpha} and \ref{c2alpha} would suffice.)

Now assume that $\dist(a,\partial\Omega)< 2R$. We set
\begin{equation}\label{eq:a27}
r_1 =
\begin{cases}
(2M)^{-1/\delta}r_0 & \text{if}\ k=1 \,,\\
(2M)^{-1}r_0 & \text{if}\ k\geq 2\,.
\end{cases}
\end{equation}
Without loss of generality we assume $M\geq \frac 12$, hence $r_1\leq r_0$. 
We will first assume that $R\leq r_1/4$, which implies that if $p\in\partial\Omega$ is chosen with $|p-a|=\dist(a,\partial\Omega)$, then
\begin{equation}
B_{2R}(a) \cap\Omega \subset B_{r_1}(p)\cap\Omega \,.
\end{equation}
(Indeed, if $|y-a|<2R$, then $|y-p|\leq |y-a|+|a-p|<2R+\dist(a,\partial\Omega)\leq 4R\leq r_1$.) Therefore, we can work in the boundary coordinates from the definition of a $C^{k,\delta}$ domain centered at the point $p$. After a translation and a rotation we may assume that $p=0$ and that there is a function $\Gamma:\{y'\in\R^{d-1}:\ |y'|<r_0\}\to\R$ with $\Gamma(0)=0$, $\nabla\Gamma(0)=0$ and
\begin{equation}
\Omega\cap B_{r_0}(0) = \{ (y',y_d)\in\R^{d-1}\times\R :\ |y'|<r_0 \,,\ y_d>\Gamma(y') \}\cap B_{r_0}(0) \,.
\end{equation}
We introduce the change of variables $\Phi:\Omega\cap B_{r_0}(0)\to\R^d_+$,
\begin{equation}
\Phi_m(y) = y_m \,\ \text{if}\ 1\leq m\leq d-1 \,,
\qquad
\Phi_d(y) = y_d-\Gamma(y') \,.
\end{equation}
The following lemma shows that decreasing $r_0$ to $r_1$ ensures that $\Phi$ is bi-Lipschitz.

\begin{lemma}\label{lem:biL}
For $x,y\in \Omega\cap B_{r_1}(0)$, we have
\begin{equation}
\frac{1}{2}|x-y| \leq \left|\Phi(x)-\Phi(y)\right| \leq \frac32 |x-y| \,.
\end{equation}
\end{lemma}

\begin{proof}
For $x,y\in\Omega\cap B_{r_0}(0)$ we have by the triangle inequality
\begin{equation}
\left| \left|\Phi(x)-\Phi(y)\right| - \left|x-y\right|\right| \leq \left|\Gamma(x')-\Gamma(y')\right|.
\end{equation}
In order to further bound this, we write, using $\nabla\Gamma(0)=0$,
\begin{equation}
\Gamma(x')-\Gamma(y')  = \int_0^1 (x'-y')\cdot(\nabla\Gamma(y'+t(x'-y'))-\nabla\Gamma(0))\,dt \,.
\end{equation}
When $k=1$, we obtain
\begin{equation}
\left| \Gamma(x')-\Gamma(y') \right| \leq M r_0^{-\delta} \int_0^1 |x'-y'| |y'+t(x'-y')|^\delta\,dt
\leq M r_0^{-\delta} \max\{|x'|,|y'|\}^\delta |x'-y'| \,.
\end{equation}
For $|x'|,|y'|\leq r_1 = (2M)^{-1/\delta} r_0$, this is $\leq |x'-y'|/2$. The argument for $k\geq 2$ is similar.
\end{proof}

Let $\tilde\Omega=\Phi(B_{2R}(a)\cap\Omega)$. This is an open set in $\R^d_+$ with a boundary portion $T=\Phi(B_{2R}(a)\cap\partial\Omega)$ on $\partial\R^d_+$. For a function $g$ on $B_{2R}(a)\cap\Omega$ we define a function $\tilde g$ on $\tilde\Omega$ by
\begin{equation}
\tilde g(x) = g(\Phi^{-1}(x)) \,.
\end{equation}
We claim that
\begin{equation}
\label{eq:equivnorm}
|\tilde g|_{k,\delta,\tilde\Omega\cup T}^{(\sigma)} \lesssim 
\sum_{j=0}^k R^{j+\sigma} \sup_{B_{2R}(a)\cap\Omega} |\partial^j g| + R^{k+\delta+\sigma} \sup_{x,y\in B_{2R}(a)\cap\Omega} \frac{|\partial^k g(x)-\partial^k g(y)|}{|x-y|^\delta}
\end{equation}
with an implicit constant depending only on $d$, $k$, $\delta$ and $M$. Indeed, by Lemma~\ref{lem:biL}, for $x\in B_{2R}(a)\cap\Omega$,
\begin{equation}
\dist(\Phi(x),\partial\tilde\Omega\setminus T) \leq \frac{3}{2} \dist(x,\partial( B_{2R}(a)\cap\Omega)\setminus(B_{2R}(a)\cap\partial\Omega)) \leq 3R\,.
\end{equation}
Moreover, for $j\leq d-1$, we have $\partial_j\tilde g = \partial_j g + \partial_d g \partial_j \Gamma$, and $\partial_d\tilde g = \partial_d g$. Since $|\partial_j \Gamma|\leq M$, we see that $|\partial \tilde g|\lesssim |\partial g|$. When computing a second derivative, also a term like $\partial_d g \partial_j\partial_k\Gamma$ appears. Bounding $|\partial_j\partial_k\Gamma|\leq M r_0^{-1}$ and $R\lesssim r_0$, we obtain $|\partial^2 \tilde g|\lesssim |\partial^2 g| + R^{-1} |\partial g|$. The arguments for higher derivatives and for the H\"older term are similar.

After these preliminaries we now return to our differential equation. We have $-\Delta u =f$ in $\Omega\cap B_{2R}(a)$ and $u=0$ on $\partial\Omega\cap B_{2R}(a)$. Therefore the functions
\begin{equation}
\tilde u(x) = u(\Phi^{-1}(x)) \,,
\qquad
\tilde f(x) = f(\Phi^{-1}(x)) \,
\end{equation}
satisfy
\begin{equation}
L \tilde u = \tilde f
\qquad\text{in}\ \tilde\Omega
\qquad\text{and}\qquad
\tilde u = 0
\qquad\text{on}\ T
\end{equation}
with the operator
\begin{equation}
L= -\sum_{r,s=1}^d \partial_r a_{r,s} \partial_s \,,
\end{equation}
where
\begin{equation}
a_{r,s}=
\begin{cases}
\delta_{r,s} & \text{if}\ r,s\leq d-1 \,. \\
1+(\nabla\Gamma)^2 & \text{if}\ r=s=d \,,\\
- \partial_r\Gamma & \text{if}\ r<d=s \,,\\
- \partial_s\Gamma & \text{if}\ s<d=r \,.
\end{cases}
\end{equation}
A  straightforward computation shows that the smallest eigenvalue of the matrix defined by $a_{r,s}$ is given by $1+\frac 12((\nabla\Gamma)^2-\sqrt{(\nabla\Gamma)^4+ 4(\nabla\Gamma)^2})$. The function $t\mapsto 1+\frac 12 (t-\sqrt{t^2+4t})$ is positive for $t\geq 0$ and strictly decreasing to $0$ as $t\to\infty$. Therefore, since $|\nabla\Gamma|\leq M$ by our definition of $C^{k,\delta}$ smoothness, we see that the lowest eigenvalue is uniformly bounded below by some $\lambda>0$ depending only on $M$.

Moreover, using the definition of a $C^{k,\delta}$-set and the fact that $R\lesssim r_0$, we deduce from \eqref{eq:equivnorm} that
\begin{equation}
\sum_{r,s} |a_{r,s}|_{k-1,\delta,\tilde\Omega\cup T}^{(0)} \leq \Lambda
\end{equation}
with $\Lambda$ depending only on $d$, $k$, $\delta$ and $M$. Similarly, for 
\begin{equation}
b_r = -\sum_{s=1}^d \partial_s a_{sr} = \begin{cases}
0 & \text{if}\ r\leq d-1 \,,\\
\Delta\Gamma & \text{if}\ r=d \,,
\end{cases}
\end{equation}
and $k\geq 2$, we have
\begin{equation}
\sum_{r} |b_{r}|_{k-2,\delta,\tilde\Omega\cup T}^{(1)} \leq \Lambda \,.
\end{equation}
From Lemmas \ref{c1alpha} and \ref{c2alpha} we conclude that
\begin{equation}
\label{eq:inequalitytilde}
\| \tilde u \|_{k,\delta,\tilde\Omega\cup T}^{(0)} \lesssim
\| \tilde u\|_{0,\tilde\Omega\cup T}^{(0)} +
\begin{cases}
\|\tilde f\|_{0,\tilde\Omega\cup T}^{(2)} & \text{if}\ k=1 \,,\\
\|\tilde f\|_{k-2,\delta,\tilde\Omega\cup T}^{(2)} & \text{if}\ k\geq 2 \,.
\end{cases}
\end{equation}
According to \eqref{eq:equivnorm}, the right side of \eqref{eq:inequalitytilde} can be further bounded by a constant (depending only on $d$, $k$, $\delta$ and $M$) times
\begin{equation}
\sup_{B_{2R}(a)\cap\Omega}|u| +
\begin{cases}
R^2 \sup\limits_{B_{2R}(a)\cap\Omega}|f| & \text{if}\ k=1 \,,\\
\sum_{j=0}^{k-2} R^{j+2} \sup\limits_{B_{2R}(a)\cap\Omega}|\partial^j f| + R^{k+\delta} \sup\limits_{x,y\in B_{2R}(a)\cap\Omega} \frac{|\partial^{k-2}f(x)-\partial^{k-2}f(y)|}{|x-y|^\delta}
& \text{if}\ k\geq 2 \,.
\end{cases}
\end{equation}
We claim that the left side of \eqref{eq:inequalitytilde} is bounded from below by a constant (depending only on $d$, $k$, $\delta$ and $M$) times
\begin{equation}
\sum_{j=0}^{k} R^{j} \sup_{B_{R}(a)\cap\Omega}|\partial^j u| + R^{k+\delta} \sup_{x,y\in B_{R}(a)\cap\Omega} \frac{|\partial^{k}u(x)-\partial^{k}u(y)|}{|x-y|^\delta} \,.
\end{equation}
The proof of the latter fact is similar to that of \eqref{eq:equivnorm}. Namely, for $x\in B_R(a)\cap\Omega$, one has
\begin{equation}
\dist(\Phi(x),\partial\tilde\Omega\setminus T) \geq \frac12 \dist(x,\partial(B_{2R}(a)\cap\Omega)\setminus(B_{2R}(a)\cap\partial\Omega)) \geq \frac12 R \,.
\end{equation}
Moreover, factors of derivatives of $\Gamma$, which appear when computing derivatives of $u$ in terms of derivatives of $\tilde u$, are handled as in the proof of \eqref{eq:equivnorm}. This  completes the proof of the theorem in case $R_0 \leq r_1/4$ with $r_1$ defined in \eqref{eq:a27}. 

The case of larger $R_0$ is readily reduced to the previous case by covering the ball $B_R(a)$ with finitely many smaller balls of size $r_1/4$. As long as $R_0/r_0$ is bounded, this only modifies the constants in the bounds. 
\hfill\qed

\section{Bounds on the Kernel of Functions of the Dirichlet Laplacian}\label{sec:appb}

In this appendix we will use the bounds in Appendix~\ref{sec:appa}, specifically Corollary~\ref{pdeboundcor}, to obtain estimates on derivatives of the integral kernel of various functions of the Dirichlet Laplacian~$\Delta_\Omega$ for $\Omega\subset \R^d$. We work in arbitrary dimension $d\geq 1$. 

\subsection{Simple Bounds}

We recall \cite[Eq. (1.9.1)]{Da} that for any $x,y\in\Omega$, one has
\begin{equation}\label{eq:FK}
0\leq e^{t\Delta_\Omega}(x,y) \leq e^{t\Delta_{\R^d}}(x,y) = (4\pi t)^{-d/2} e^{-(x-y)^2/(4t)} \,.
\end{equation}
Therefore, by Bernstein's theorem we infer that for any completely monotone function $f$ on $[0,\infty)$, we have
\begin{equation}\label{eq:b2}
0\leq f(-\Delta_\Omega)(x,y) \leq f(-\Delta_{\R^3})(x,y) = \int_{\R^d} f(k^2) e^{ik\cdot(x-y)} \frac{dk}{(2\pi)^d} \,.
\end{equation}
This bound is used in the main text multiple times, for instance with $f(t)=t^{-1} e^{-t/K^2}$ and $f(t)=(t+K^2)^{-3}$.

To motivate the following, we shall first derive a more general but slightly worse bound on the diagonal $x=y$, assuming only that $f$ is non-increasing. Assuming that $\Omega$ is bounded (or more generally that the spectrum of $-\Delta_\Omega$ is discrete) we shall denote the eigenvalues of $-\Delta_\Omega$ (in increasing order and repeated according to their multiplicities) by $e_n$, and the corresponding eigenfunctions by $\varphi_n$.  According to \eqref{eq:FK} we have for any $K>0$
\begin{equation}
\sum_{e_n\leq K^2} |\varphi_n(x)|^2 \leq e^{t K^2} e^{t\Delta_\Omega}(x,x) \leq   e^{tK^2} (4\pi t)^{-d/2} \,.
\end{equation}
Optimizing in $t$ yields
\begin{equation}
\sum_{e_n\leq K^2} |\varphi_n(x)|^2 \leq \left( \frac e{2\pi d} \right)^{d/2} K^d = \left( \frac{2e}d \right)^{d/2} \Gamma(1+d/2)  \int_{\{|k|\leq K\}} \frac{dk}{(2\pi)^d} \,.
\end{equation}
Any non-increasing function $f$ with $\lim_{t\to\infty} f(t)=0$ can be written as a superposition of characteristic functions as $f(t)= -  \int_0^\infty \chi_{\{t\leq s\}} f'(s)\,ds$, and hence
\begin{equation}\label{eq:b5}
 \sum_n f(e_n) |\varphi_n(x)|^2 = f(-\Delta_\Omega)(x,x) \leq \left( \frac{2e}d \right)^{d/2} \Gamma(1+d/2)   \int_{\R^d} f(k^2) \frac{dk}{(2\pi)^d}
\end{equation}
for non-increasing functions. 

\subsection{Bounds on the Diagonal}

We now use the same method to derive bounds on $\sum_n f(e_n) |\partial^\beta\varphi_n(x)|^2$. To do so we shall use Corollary~\ref{pdeboundcor} to prove the following.

\begin{lemma}\label{lem:b1}
Assume that $\Omega\subset \R^d$ is a bounded, open $C^{k,\delta}$ set for some $k\geq 1$ and $0<\delta<1$, and let $R_0>0$. For any bounded function $g:\R_+\to \R$ of compact support,  any $\beta\in\N_0^d$ with $|\beta|\leq k$ and any $R\in (0,R_0)$, 
\begin{equation}\label{eq:b6}
R^{2|\beta|} \sum_n g(e_n)^2 |\partial^\beta \varphi_n(x)|^2 \lesssim \sum_{j=0}^{|\beta|}  \sup_{x'\in B_R(x)\cap\Omega} \sum_n g(e_n)^2 (R^2 e_n)^{2j} |\varphi_n(x')|^2 
\end{equation}
for all $x\in\Omega$.
\end{lemma}

\begin{proof}
We proceed by induction in $|\beta|$. For $|\beta|=0$, \eqref{eq:b6} obviously holds. Assume now $|\beta|\geq 1$. Pick a $\psi \in L^2(\Omega)$, and let $u= g(-\Delta_\Omega) \psi$. From Corollary~\ref{pdeboundcor}, we obtain for any $x\in \Omega$ 
\begin{equation}\label{eq:b7}
R^{|\beta|} | \partial^\beta u(x)| \lesssim \sup_{B_R(x)\cap \Omega} |u(x')| + \sum_{\alpha:\, |\alpha|<|\beta|} R^{|\alpha|+2} \sup_{B_R(x)\cap\Omega} | \partial^\alpha \Delta_\Omega u(x') |\,.
\end{equation}
Now
\begin{equation}
|u(x')| = | g(-\Delta_\Omega) \psi(x')| = \left| \sum_n g(e_n) \langle\varphi_n| \psi\rangle \varphi_n(x') \right| \leq \left( \sum_n g(e_n)^2 |\varphi_n(x')|^2\right)^{1/2} \|\psi\|_2
\end{equation}
and similarly
\begin{equation}\label{eq:b9}
| \partial^\alpha \Delta_\Omega u(x') | \leq \left( \sum_n g(e_n)^2 e_n^2 |\partial^\alpha\varphi_n(x')|^2\right)^{1/2} \|\psi\|_2 \,.
\end{equation}
By combining \eqref{eq:b7}--\eqref{eq:b9} and using the induction hypotheses for $\alpha$ with $|\alpha|< |\beta|$, we therefore obtain the bound
\begin{equation}
R^{2|\beta|} |\partial^\beta g(-\Delta_\Omega) \psi(x)|^2 \lesssim \|\psi\|_2^2 \sum_{j=0}^{|\beta|} \sup_{B_{2R}(x)\cap\Omega} \sum_n g(e_n)^2 (R^2e_n)^{2j} |\varphi_n(x')|^2
\end{equation}
valid for all $\psi \in L^2(\Omega)$. 
Since 
\begin{equation}
\sup_\psi \|\psi\|_2^{-2} |\partial^\beta g(-\Delta_\Omega) \psi(x)|^2 =\sum_n g(e_n)^2  |\partial^\beta\varphi_n(x)|^2
\end{equation}
the result follows. 
\end{proof}

We apply \eqref{eq:b6} with $g$ the characteristic function of $\{e\leq K^2\}$ for some $K>0$,  $R=K^{-1}$ and $R_0=e_1^{-1/2}$. This yields
\begin{equation}\label{eq:b12}
\sum_{e_n\leq K^2} | \partial^\beta \varphi_n(x)|^2 \lesssim K^{2|\beta|} \sup_{B_{K^{-1}}(x)\cap\Omega} \sum_{e_n\leq K^2} |\varphi_n(x')|^2 \lesssim K^{2|\beta| + d}
\end{equation}
where we have used \eqref{eq:b5} in the last step. 
More generally, we obtain for any non-increasing function $f$ with $\lim_{t\to\infty} t^{d/2+|\beta|}f(t)=0$  that
\begin{align}\nonumber
\sum_{n} f(e_n) |\partial^\beta\varphi_n(x)|^2 & = - \int_0^\infty \sum_{e_n\leq E} |\partial^\beta\varphi_n(x)|^2 f'(E)\,dE \\ \nonumber 
& \lesssim - \int_0^\infty E^{d/2+|\beta|} f'(E) \,dE \\ \nonumber
& = \const \int_0^\infty E^{d/2+|\beta|-1} f(E)\,dE \\
& =\const   \int_{\R^d} k^{2|\beta|} f(k^2) \frac{dk}{(2\pi)^d} \,. \label{eq:b13}
\end{align}

The validity of \eqref{eq:b12} is shown in  \cite[Thm.~ 17.5.3]{Ho} 
 if $\Omega$ has $C^\infty$ boundary. Following the proof there (which is based on regularity theory in $L^2$-based Sobolev spaces) one sees that a certain finite number of derivatives is actually sufficient, but the result is not as precise as ours, which only requires $C^{|\beta|,\delta}$ regularity of the boundary.

\subsection{Offdiagonal Bounds}\label{ss:b3}

In this section we shall derive a bound on the derivatives of the kernel of certain functions of the Dirichlet Laplacian, valid even away from the diagonal. These bounds are much less general than the ones in the previous two subsections, however. For simplicity we only consider the particular class of functions needed in the main text, but the method obviously extend to other functions as well. 

For $\Lambda>0$ and $\ell > 0$, let 
\begin{equation}
z_\ell(t) = t^{-\ell} \left( 1 - e^{-t/\Lambda^2} \right)^2.
\end{equation}

\begin{lemma}
Assume that $\Omega\subset \R^d$ is a bounded, open $C^{k,\delta}$ set for some $k\geq 1$ and $0<\delta<1$. For  any $\beta\in\N_0^d$ with $|\beta|\leq k$ and $|\beta| < 2+d/2$, and any $\ell \in ( |\beta|, 2+d/2)$ and $\Lambda >0$,   we have
\begin{equation}\label{a4:5}
 \left| \partial_x^\beta  z_\ell (-\Delta_\Omega)(x,y) \right| \lesssim 
 \begin{cases}   
 |x-y|^{2\ell -d -|\beta|}  & \text{for $\ell < d/2$} \\
   \ln(1+(\Lambda|x-y|)^{-1})  |x-y|^{-|\beta|}& \text{for $\ell = d/2$} \\
   \Lambda^{d-2\ell} |x-y|^{-|\beta|} & \text{for $\ell > d/2$}
 \end{cases}
\end{equation}
for $\Lambda|x-y|\leq 1$, and 
\begin{equation}\label{a4:5a}
 \left| \partial_x^\beta  z_\ell (-\Delta_\Omega)(x,y) \right| \lesssim  \Lambda^{-4} |x-y|^{2\ell - 4-d-|\beta|}
 \end{equation}
for $\Lambda|x-y| \geq 1$.
\end{lemma}

\begin{proof}
We use again Corollary~\ref{pdeboundcor} above. A simple induction argument as in the proof of Lemma~\ref{lem:b1} shows that 
\begin{equation}\label{a4:6}
R^{|\beta|}  \left| \partial_x^\beta   z_\ell (-\Delta_\Omega)(x,y) \right| \lesssim \sum_{i=0}^{|\beta|}  R^{2i} \sup_{x'\in B_R(x)} \left|  z_{\ell -i} (-\Delta_\Omega)(x',y) \right| 
\end{equation}
for any $R>0$ (smaller than some arbitrary, fixed value). To estimate the right side of \eqref{a4:6}, we write  for $j>0$
\begin{align}\nonumber
z_{j}(t) & =  t^{ -j } \left( 1 - e^{-t/\Lambda^2} \right)^2 \\ &  = \frac 1{\Gamma(j )} \int_0^\infty e^{-\lambda t}  \left( \lambda^{j-1} - 2 \left[ \lambda- \Lambda^{-2}\right]_+^{j-1} +  \left[ \lambda- 2 \Lambda^{-2}\right]_+^{j-1}  \right) d\lambda
\end{align}
where the term $\left[ \lambda- \Lambda^{-2}\right]_+^{j-1}$ is understood as being zero for $\lambda < \Lambda^{-2}$ even when $j < 1$, and likewise for $ \left[ \lambda- 2 \Lambda^{-2}\right]_+^{j-1}$. In particular, from \eqref{eq:FK}, we thus have
\begin{equation}\label{a4:8} 
\left| z_{j}(-\Delta_\Omega)(x,y) \right|  \leq \Lambda^{d-2j  } f_j( \Lambda |x-y|)
\end{equation}
with
\begin{align}
f_j(t) &  = \frac {1}{\Gamma(j) (4\pi)^{d/2} } \int_0^\infty e^{- t^2/(4\lambda)}  \left| \lambda^{j-1} - 2  \left[ \lambda- 1\right]_+^{j-1} +  \left[ \lambda- 2 \right]_+^{j-1}  \right| \lambda^{-d/2} \,d\lambda \,.
\end{align}
We note that 
\begin{equation}
 \left| \lambda^{j-1} - 2\left[ \lambda- 1\right]_+^{j-1} + \left[ \lambda- 2 \right]_+^{j-1}  \right| \lesssim \lambda^{j-3} 
\end{equation}
for $\lambda \geq 3$. Using this, one readily checks that as long as $0<j<2+d/2$, 
\begin{equation}\label{a4:11}
f_j(t) \lesssim t^{2j-4-d} \quad \text{for $t \geq 1$, and \ } f_j (t) \lesssim \begin{cases} 1 & \text{for $j>d/2$} \\ \ln (2/t) & \text{for $j=d/2$} \\ t^{2j-d}  & \text{for $j<d/2$} \end{cases} \quad \text{ for $t\leq 1$} \,.
\end{equation}
We plug these bounds into \eqref{a4:8} and choose $R= |x-y|/2$ in \eqref{a4:6}. (Note that $R\leq R_0$, as required for \eqref{a4:6}, where $R_0=$ diameter of $\Omega$.)   
For all $x' \in B_R(x)$, we then have $ |x'-y|  \geq  |x-y|/2$, and hence \eqref{a4:6}, \eqref{a4:8} and \eqref{a4:11} imply the desired bounds \eqref{a4:5} and \eqref{a4:5a} for this choice of $R$. 
\end{proof}

Recall the definition $u_{jk}(x) = \sup_{y\in \R^3} | p_j p_k |p|^{-4} w_{x+y}(y)|$ with
\begin{equation}\label{def:w2}
w_x(y) = z_{1/2}(-\Delta_\Omega)(x,y) \,.
\end{equation}
Applying the bounds \eqref{a4:5} and \eqref{a4:5a}, with $\ell = 5/2$, $d=3$ and $|\beta|=2$, we readily obtain
\begin{equation}\label{a4:1}
u_{jk}(x) \lesssim \min \{ \Lambda^{-2} |x|^{-2} , \Lambda^{-4} |x|^{-4} \}\,.
\end{equation}
The function $\min\{|x|^{-2}, |x|^{-4}\}$ is in $L^{6/5}(\R^3)$ and hence has finite Coulomb norm. By the Hardy--Littlewood--Sobolev inequality and scaling, it thus follows immediately that 
$\| u_{jk}\|_{\rm C} \lesssim \Lambda^{-5/2}$, 
as claimed in \eqref{claim1}. 

\bigskip
\noindent {\bf Acknowledgments.} 
Partial support through  the U.S. National Science Foundation, grant DMS-1363432 (R.L.F.), and the European Research Council (ERC) under the European Union's Horizon 2020 research and innovation programme (grant agreement No 694227; R.S.),  is acknowledged.


\end{document}